\documentclass[English, 12pt]{article}

\usepackage[utf8]{inputenc}
\usepackage{amsmath,amssymb, mathrsfs, dsfont, amsthm, color}
\usepackage[margin=1in]{geometry}
\usepackage{tikz}
\usepackage{mathtools}
\usepackage{enumitem}
\usepackage{soul} 

\usepackage{amsxtra}
\usepackage{stmaryrd}

\usepackage[title]{appendix}

\usepackage[square,numbers]{natbib}
\usepackage{titlesec}

\usepackage{amsfonts}
\usepackage{amsthm}
\usepackage{amsmath}
\usepackage{amssymb}

\usepackage{bbm}
\usepackage{dsfont}
\newtheorem{thm}{Theorem}[section]
\newtheorem{prop}[thm]{Proposition}

\newtheorem{lemma}[thm]{Lemma}

\theoremstyle{definition}
\newtheorem{definition}[thm]{Definition}
\newtheorem{remark}[thm]{Remark}
\newtheorem{assumption}[thm]{Assumption}
\newtheorem{example}[thm]{Example}

\newcommand{\bt}{\begin{thm}}
\newcommand{\et}{\end{thm}}
\newcommand{\br}{\begin{remark}}
\newcommand{\er}{\end{remark}}
\newcommand{\bl}{\begin{lemma}}
\newcommand{\el}{\end{lemma}}
\newcommand{\bp}{\begin{proof}}
\newcommand{\ep}{\end{proof}}
\newcommand{\bal}{\begin{align*}}
\newcommand{\eal}{\end{align*}}
\newcommand{\bi}{\begin{itemize}}
\newcommand{\be}{\begin{equation}}
\newcommand{\ee}{\end{equation}}
\newcommand{\bea}{\begin{eqnarray}}
\newcommand{\eea}{\end{eqnarray}}
\newcommand{\ba}{\begin{align*}}
\newcommand{\ea}{\end{align*}}
\newcommand{\ei}{\end{itemize}}

\newcommand{\R}{\mathbb{R}}

\newcommand{\cF}{\mathcal{F}}

\newcommand{\cA}{\mathcal{A}}

\newcommand{\E}{\mathrm{E}}

\numberwithin{equation}{section}

\newcommand{\diff}{\mathrm{d}}

\begin{document}
\title{Robust utility maximisation under proportional transaction costs for c{\`a}dl{\`a}g price processes\footnote{We thank Huy Chau for his comments on an earlier version of the paper.}}
\author{Christoph Czichowsky\footnote{Department of Mathematics, London School of Economics and Political Science, Columbia House, Houghton Street, London WC2A 2AE, UK, {\tt c.czichowsky@lse.ac.uk}.} 
\hspace{20pt}Raphael Huwyler\footnote{Department of Mathematical and Statistical Sciences, University of Alberta, Edmonton AB T6G 2N8, Canada, {\tt huwyler@ualberta.ca}.}
\\
\\
\today.
}
\date{}
\maketitle

\begin{abstract}
\noindent
We consider robust utility maximisation in continuous-time financial markets with proportional transaction costs under model uncertainty. For this purpose, we work in the framework of \citet{chau2019}, where robustness is achieved by maximising the worst-case expected utility over a possibly uncountable class of models that are all given on the same underlying filtered probability space. In this setting, we give sufficient conditions for the existence of an optimal trading strategy, extending the result for utility functions on the positive half-line of \citet{chau2019} from continuous to general strictly positive c{\`a}dl{\`a}g price processes and from complete to incomplete filtrations. Our result allows us to provide a positive answer to an open question pointed out in \citet{chau2019}, and shows that the embedding into a countable product space is not essential.
\end{abstract}

\section{Introduction}
Maximising the expected utility from terminal wealth is a classical problem in Mathematical Finance and Financial Economics. Recently, there has been a lot of interest in robust utility maximisation under model uncertainty, where one maximises the worst-case expected utility over a class of models. The motivation for this viewpoint is that the resulting trading strategies are less sensitive to changes of the underlying model and in this sense more robust to model misspecification. 

In this paper, we consider robust utility maximisation under proportional transaction costs in the framework of model uncertainty of \citet{chau2019}. Here, the worst-case expected utility over a possibly uncountable class of models, which are all given on the same underlying filtered probability space, is maximised. In this setting, we extend the existence result of \citet{chau2019} for utility functions on the positive half-line from continuous to general strictly positive c\`adl\`ag price processes, and from complete to incomplete filtrations. This work answers an open question in \citet{chau2019}. Without frictions, a rich amount of examples of discrete-time models for this setting of model uncertainty has been proposed in \citet{rasonyi2021}. Our results cover the corresponding models under proportional transaction costs and allow to consider both discrete- and continuous-time models in a unified framework.
 
As already explained in \citet{chau2019}, most of the literature typically parameterises model uncertainty by a family $\mathbb{P}$ of probability measures that are given on some underlying canonical probability space (see, e.g.,~\citet{biagini2017}, \citet{lin2021} and \citet{neufeld2018}). This means that the dynamics of the risky asset is given by a fixed process and model uncertainty is described by a family of different distributions of this process. The discrete-time case with unbounded utility functions defined on $(0,\infty)$ has been solved in \citet{blanchard2018}. In the context of transaction costs, the discrete-time case with transaction costs and exponential utility preferences has been treated in \citet*{deng2020}. For diffusion models, drift uncertainty can be modelled by considering a family of absolutely continuous measures that are dominated by a single measure $P^{\ast}$ (see e.g., \citet{quenez2004} and \citet{schied2006}), while the case of volatility uncertainty requires an uncountable family of singular measures (see e.g., \citet{denis2013}).

In contrast to the approach above, \citet{chau2019} propose a setup of model uncertainty, where different stock price processes are considered. That is, they suggest to work on a fixed filtered probability space and to use a family of stochastic processes $S^\theta$ indexed by $\theta$ in a non-empty set $\Theta$ to describe model uncertainty. From a mathematical point of view, the main advantage of this setup is that no topological or measurability assumptions are needed on the set of parameters $\Theta$ representing the different models. In contrast, the parametrisation via a family of measures $\mathbb{P}$ incorporates technical issues such as the treatment of null events and filtration completion (see e.g., \citet{biagini2017b}, \citet{bouchard2015} and \citet{nutz2016}). However, since typical examples consider an uncountable set of models $\Theta$, one cannot complete the filtration with respect to the null sets arising from any price process $S^\theta$ and has to work with incomplete filtrations. Working with filtrations under ``unusual'' conditions, that is, without the usual condition of completeness, then brings its own challenges. This additional difficulty is in contrast to \citet{chau2019}, who still assume the filtration to be complete. Besides this fact, we refer to the original paper of \citet{chau2019} for a more detailed comparison between these two approaches to model uncertainty.

Within their setting of model uncertainty, \citet{chau2019} observe that, similarly as in the case of a single model in \citet{guasoni2002}, it is more suitable to optimise directly over trading strategies, and hence stochastic processes, rather than over terminal wealths given by random variables as in the classical case with only one price process. For this purpose, they  exploit that, for continuous price processes, it is sufficient to model trading strategies by c\`adl\`ag finite variation processes under proportional transaction costs. The key insight is that c\`adl\`ag finite variation processes can be identified with their values along the rational numbers and hence objects that are taking values in the countable product of complete metric spaces. Nevertheless, this approach does not avoid the measurability problem arising from working with uncountably many models simultaneously, if one does not assume the filtration to be complete.

While c\`adl\`ag strategies are sufficient to obtain the optimal strategies for continuous price processes under proportional transaction costs, this fact is no longer true for price processes with jumps; see Example 4.1 in \citet{czichowsky2016}. For c{\`a}dl{\`a}g price processes, it matters, whether one is trading immediately before, just at, or immediately after a jump. Therefore, trading strategies have to be modelled by general predictable finite variation processes that can have left and right discontinuities and can no longer be identified with their values along the rationals. To overcome this difficulty, we use a version of Helly's Theorem of \citet{campi2006}. The latter allows us to obtain a sequence of trading strategies that converges to a c\`adl\`ag modification of a suitable limit process at all time points except for the discontinuity points of that modified limit process. The key is now that we can show that the set, where the convergence can fail, is the same for all models. Hence, it can be exhausted by countably many stopping times. The stopping times can be obtained by using a suitable version of the Debut Theorem for filtrations that are not complete. This insight allows us to achieve the convergence at these points as well by a diagonal sequence argument. Mathematically speaking, while \citet{chau2019} work with the topology of $P$-a.s. convergence along all rational times on the set of c{\`a}dl{\`a}g finite variation processes, we work with the topology of convergence in probability at all $[0,T]$-valued stopping times. Somewhat surprisingly, our result shows that the embedding into a countable product metric space indexed by the rationals as used in \citet{chau2019} is not essential.

We assume that our utility functions have a reasonable asymptotic elasticity as in the classical single model framework of \citet{kramkov1999}. This assumption allows us to obtain the optimal trading strategy by directly optimising in the primal problem of maximising expected utility from terminal wealth, and we do not need specific properties of the dual problem. Therefore, we only need the existence of (locally) consistent price systems for \emph{one} level of transaction costs $\lambda'\in(0,\lambda)$ rather than for all $\lambda'\in(0,\lambda)$ as in \citet{chau2019}; see Remark 4.5 of \citet{chau2019} and Lemma~\ref{lem:bipolar_relation_model_based} below.

In the framework of model uncertainty of \citet{chau2019}, a super-replication theorem has been recently established in \citet{chau2022}. For model uncertainty with uncountably many probability measures on the same probability space, \citet{bartl2021} derived a duality result, in the spirit to the one of \citet{kramkov1999}, for the case of a single model, for utility maximisation from terminal wealth without transaction costs.

The paper is organised as follows. We introduce the setting and formulate the problem in Section~\ref{sec:setup}. Our main result is stated and explained in Section~\ref{sec:main_result}. The proof of the main result is covered in Section~\ref{sec:proof_main_result}. For better readability, some proofs and explanations are deferred to Appendix~\ref{app:proof_con_compct}.

\section{Formulation of the problem}\label{sec:setup}
We consider a financial market consisting of one riskless asset and one risky asset. The riskless asset has constant price $1$. The dynamics of the risky asset is uncertain. For this purpose, let $\Theta$ be a non-empty set and consider a family of strictly positive adapted c{\`a}dl{\`a}g processes $(S_{t}^{\theta})_{0\leq t\leq T}$, for each $\theta\in\Theta$, on some underlying filtered probability space $(\Omega,\mathcal{F},(\mathcal{F}_{t})_{0\leq t\leq T},P)$. No further conditions are imposed on $\Theta$. By passing to the right-continuous version, we can assume without loss of generality that the filtration $\mathbb{F}\coloneqq(\mathcal{F}_{t})_{0\leq t\leq T}$ is right continuous. However, the filtration does not necessarily need to be complete. We denote by $\mathcal{F}^{P}$ the $P$-completion of $\mathcal{F}$, and accordingly $\mathbb{F}^{P}$ the usual augmentation of $\mathbb{F}$ by adjoining all the $P$-negligible sets to $\mathcal{F}_{0}$. In this framework, model uncertainty is then incorporated by considering simultaneously all models $(S_{t}^{\theta})_{0\leq t\leq T}$, for $\theta\in\Theta$, as possible evolutions for the price process of the risky asset, so that the set $\Theta$ provides the parametrisation of model uncertainty.

To illustrate why this particular setting of model uncertainty is useful, \citet{chau2019} provide some interesting examples of well-known models in the finance literature (see Examples~2.2, 2.4 and 2.5 therein). For instance, in Example~2.2, they describe a robust version of the Black-Scholes model, where the parametrisation of model uncertainty is a subset of $\mathbb{R}^{2}$, which describes potential values for the drift and volatility term of the geometric Brownian motion that models the price process. To see why such a parametrisation of robustness is useful and interesting in our extended setup that includes c{\`a}dl{\`a}g price processes, we provide some more examples.
\begin{example}\label{ex:robust_Merton}
In the robust Merton jump diffusion model (cf., \citet{merton1976}), the price of the risky asset satisfies the stochastic differential equation
\begin{align*}
\diff S_{t}^{(\mu,\sigma,\nu)}&=S_{t}^{(\mu,\sigma,\nu)}((\mu-\nu k)\diff t+\sigma\diff W_{t}+\diff X_{t}^{\nu}),\qquad S_{0}^{(\mu,\sigma,\nu)}=s_{0}>0,
\end{align*}
where $\mu$ and $\sigma$ are constants, $W$ is a standard Brownian motion, and $X^{\nu}$ is an independent compound Poisson process with intensity $\nu>0$, whose jump size distribution has expected value $k$. In particular, the process $X^{\nu}=(X_{t}^{\nu})_{t\geq 0}$ satisfies
\begin{equation}\label{eq:compound_Poisson_process}
X_{t}^{\nu}\coloneqq\sum_{i=1}^{N_{t}^{\nu}}Y_{i}^{\nu},\qquad t\geq 0,
\end{equation}
where $N^{\nu}=(N_{t}^{\nu})_{t\geq 0}$ is a Poisson process with intensity $\nu$, the random variables $Y_{i}^{\nu}$, $i\in\mathbb{N}$, are independent and identically distributed with mean $\E[Y_{i}^{\nu}]=k$, and which are independent of the process $N^{\nu}$. The uncertainty is then modelled by
\begin{align*}
\Theta&=\big\{\theta=(\mu,\sigma,\nu)\in\mathbb{R}^{3}\colon\underline{\mu}\leq\mu\leq\overline{\mu},\,\underline{\sigma}\leq\sigma\leq\overline{\sigma},\,\underline{\nu}\leq\nu\leq\overline{\nu}\big\},
\end{align*}
where $\underline{\mu}\leq\overline{\mu}$, $0<\underline{\sigma}\leq\overline{\sigma}$ and $0<\underline{\nu}\leq\overline{\nu}$ are given constants. The classical, non-robust Merton jump diffusion model corresponds to the case where $\underline{\mu}=\overline{\mu}$, $\underline{\sigma}=\overline{\sigma}$ and $\underline{\nu}=\overline{\nu}$. In fact, if we set $\nu=0$, we obtain the non-robust Black-Scholes model as described in Example~2.2 in \cite{chau2019}.
\end{example}
On one hand, Example~\ref{ex:robust_Merton} shows the importance of considering c{\`a}dl{\`a}g price processes, as very general models may include jump terms to describe the ``abnormal'' variations in the price, as \citet{merton1976} describes it. These variations are due to the arrival of important new information that has more than a marginal effect on the price. On the other hand, we may see from Example~\ref{ex:robust_Merton} how this particular way of parameterising model uncertainty can be useful in statistical estimation and model calibration, and therefore relevant for practical applications.
\begin{example}\label{ex:compound_Poisson_log_price_moves}
The compound Poisson process can be used to describe more abstract models with jumps. In particular, we may include the parameters of the jump size distribution in the parametrisation of model uncertainty. For instance, we can define the price process as
\begin{align*}
\log(S_{t}^{(\theta_{1},\dots,\theta_{k},\nu)}/S_{0}^{(\theta_{1},\dots,\theta_{k},\nu)})&=X_{t}^{1,(\theta_{1},\dots,\theta_{k},\nu)}-X_{t}^{2,(\theta_{1},\dots,\theta_{k},\nu)},\qquad k\in\mathbb{N},
\end{align*}
where $X^{1,(\theta_{1},\dots,\theta_{k},\nu)}$ and $X^{2,(\theta_{1},\dots,\theta_{k},\nu)}$ are two independent copies of the compound Poisson process satisfying equation \eqref{eq:compound_Poisson_process} (see Section~3 in \cite{geman2001} for details), the random variables $Y_{i}^{(\theta_{1},\dots,\theta_{k},\nu)}$, $i\in\mathbb{N}$, are distributed with specific parameters $\theta_{1},\dots,\theta_{k}$, and the Poisson process $N^{\nu}$ has intensity $\nu>0$. For instance, we could take $Y_{i}^{(\sigma^2,\nu)}=\vert Z_{i}^{(\sigma^2,\nu)}\vert$, where each $Z_{i}^{(\sigma^2,\nu)}$ satisfies $Z_{i}^{(\sigma^2,\nu)}\sim\mathcal{N}(0,\sigma^{2})$, so that $\theta_{1}=\sigma^{2}$. Another example is to take $Y_{i}^{(\alpha,\nu)}\sim\mathrm{Exp}(\alpha)$, so that $\theta_{1}=\alpha$.
\end{example}
\begin{example}
    The previous example is a specific case of a larger class of models (see \citet{geman2001}), where the price process is a time-changed Brownian motion. The underlying idea uses the result that every semimartingale can be written as a time-changed Brownian motion on some adequately defined probability space (see, e.g., \citet{monroe1978}). If the time change was continuous, then under fairly general conditions it could be represented as an It{\^o} process. Since time is increasing, this continuity assumption would force us to set the volatility term of the time change process to zero, which means that the time change is locally deterministic. If we relate the time-change to the information flow embedded in the price, where there is local uncertainty about these flows, we must assume that the time change is a purely discontinuous process. Hence, the price process itself is also purely discontinuous.

    For instance, in Example~\ref{ex:compound_Poisson_log_price_moves}, if we choose 
    \begin{align*}
        \log(S_{t}^{(\sigma,\nu)}/S_{0}^{(\sigma,\nu)})&=X_{t}^{1,(\sigma,\nu)}-X_{t}^{2,(\sigma,\nu)},\qquad X_{t}^{j,(\sigma,\nu)}=\sum_{i=1}^{N_{t}^{j,\nu}}Y_{i}^{j,(\sigma,\nu)},\quad j\in\{1,2\},
    \end{align*}
    where $N^{1,\nu}$ and $N^{2,\nu}$ are independent Poisson processes with intensity $\nu>0$, and where $Y_{i}^{j,(\sigma,\nu)}$, $i\in\mathbb{N}$ and $j\in\{1,2\}$, are independent random variables with a half-normal distribution (i.e., $Y_{i}^{j,(\sigma,\nu)}=\vert Z_{i}^{j,(\sigma,\nu)}\vert$ with $Z_{i}^{j,(\sigma,\nu)}\sim\mathcal{N}(0,\sigma^{2})$), the time change is $T^{\nu}(t)=N_{t}^{1,\nu}+N_{t}^{2,\nu}$, so that
    \begin{align*}
        \log(S_{t}^{(\sigma,\nu)}/S_{0}^{(\sigma,\nu)})&=\sigma W_{T^{\nu}(t)}=\sigma W_{N_{t}^{1,\nu}+N_{t}^{2,\nu}},\qquad t\geq 0,
    \end{align*}
    where $W=(W_{t})_{t\geq 0}$ is a standard Brownian motion. To avoid measurability issues, one takes different Brownian motions for each model. For the details of this derivation, we refer to Section~3.2 in \cite{geman2001}. In such a model, the parametrisation of model uncertainty is of the form $\Theta=\{\theta=(\sigma,\nu)\in\mathbb{R}^{2}\colon 0<\underline{\sigma}\leq\sigma\leq\overline{\sigma},\,0<\underline{\nu}\leq\nu\leq\overline{\nu}\}$.

    Another model class is to use the same price structure as above, where the log-price is the difference of two independent pure jump processes, but instead of using the compound Poisson process, we may use a Gamma process. In particular, let
    \begin{align*}
        \log(S_{t}^{(\mu_{1},\nu_{1},\mu_{2},\nu_{2})}/S_{0}^{(\mu_{1},\nu_{1},\mu_{2},\nu_{2})})&=U_{t}^{1,(\mu_{1},\nu_{1})}-U_{t}^{2,(\mu_{2},\nu_{2})},\quad U_{t}^{j,(\mu_{j},\nu_{j})}=\frac{\nu_{j}}{\mu_{j}}\gamma_{j}\bigg(\frac{\mu_{j}^{2}}{\nu_{j}}t\bigg),\>j\in\{1,2\},
    \end{align*}
    where $\gamma_{1}$ and $\gamma_{2}$ are two independent standard Gamma processes, so that $\gamma_{j}(t)\sim\mathrm{Gamma}(t,1)$ for $j\in\{1,2\}$, and therefore $U_{t}^{j,(\mu_{j},\nu_{j})}\sim\mathrm{Gamma}(\frac{\mu_{j}^{2}}{\nu_{j}}t,\frac{\mu_{j}}{\nu_{j}})$. More precisely, the increments $U_{t+h}^{j,(\mu_{j},\nu_{j})}-U_{t}^{j,(\mu_{j},\nu_{j})}$, for any $h>0$, have density
    \begin{align*}
        f_{U_{t+h}^{j,(\mu_{j},\nu_{j})}-U_{t}^{j,(\mu_{j},\nu_{j})}}(x)&=\bigg(\frac{\mu_{j}}{\nu_{j}}\bigg)^{\frac{\mu_{j}^{2}h}{\nu_{j}}}\frac{x^{\frac{\mu_{j}^{2}h}{\nu_{j}}-1}\exp\big(-\frac{\mu_{j}}{\nu_{j}}x\big)}{\Gamma\big(\frac{\mu_{j}^{2}h}{\nu_{j}}\big)},\qquad x>0,
    \end{align*}
    so that $U_{t+h}^{j,(\mu_{j},\nu_{j})}-U_{t}^{j,(\mu_{j},\nu_{j})}$ has mean $\mu_{j}h$ and variance $\nu_{j}h$. The parameters $\mu_{j}$ and $\nu_{j}$ are thus called mean and variance rates, respectively. Under suitable assumptions on the parameters, it turns out that the log-price process is of the form $\alpha\gamma_{3}(t)+\beta W_{\gamma_{3}(t)}$, where $\alpha$ and $\beta$ are constants, $\gamma_{3}(t)$ is a linear combination of the gamma processes $\gamma_{1}(t)$ and $\gamma_{2}(t)$, and where $W=(W_{t})_{t\geq 0}$ is a standard Brownian motion. Again, each model has its own Brownian motion. For details, we refer to Section~3.4 in \cite{geman2001} (see also \cite{madan1998} for more information on the variance-Gamma process). With respect to the parametrisation of model uncertainty, we use here
    \begin{align*}
        \Theta&=\big\{\theta=(\mu_{1},\nu_{1},\mu_{2},\nu_{2})\in\mathbb{R}^{4}\colon0<\underline{\mu_{j}}\leq\mu_{j}\leq\overline{\mu_{j}},\,0<\underline{\nu_{j}}\leq\nu_{j}\leq\overline{\nu_{j}},\,j\in\{1,2\}\big\}.
    \end{align*}

    Similar models using this particular approach of time change also include general subordinators (see \cite{geman2001} for details). These models are relevant because they are capable of modelling the local uncertainty of the information flow that influences the price of the risky asset, and therefore require jump processes. On the other hand, the described models are tailor-made for statistical inference and model calibration, and therefore contain an important component for practical applications.
\end{example}
In each model, we assume that the risky asset $S^{\theta}$ can be traded under proportional transaction costs $\lambda\in(0,1)$. That is, an agent can buy the risky asset at  the higher \textit{ask price} price $S^{\theta}$ but can only sell it at the lower \textit{bid price} $(1-\lambda)S^{\theta}$. The riskless asset can be traded without transaction costs.

Trading strategies are given by $\mathbb{R}^{2}$-valued, $\mathbb{F}$-predictable processes $H=(H_{t}^{0},H_{t}^{1})_{0\leq t\leq T}$, whose total variation $\vert H\vert=(\vert H^{0}\vert_{t},\vert H^{1}\vert_{t})_{0\leq t\leq T}$ is a $[0,\infty]\times[0,\infty]$-valued, $\mathbb{F}$-predictable process satisfying $P[\vert H^{0}\vert_{T}<\infty]=P[\vert H^{1}\vert_{T}<\infty]=1$. For each real-valued, predictable process $H$, whose total variation is $P$-a.s.~finite, we define $\Delta H_{t}\coloneqq H_{t}-H_{t-}$, and $\Delta_{+}H_{t}\coloneqq H_{t+}-H_{t}$, where $H_{t-}:=\lim_{s\uparrow t}H_s$ and $H_{t+}:=\lim_{s\downarrow t}H_s$ denote the left and right limits, respectively, and the processes
\begin{align*}
    H_{t}^{d}\coloneqq\sum_{s\leq t}\Delta H_{s},\quad\text{and}\quad H_{t}^{d,+}\coloneqq\sum_{s<t}\Delta_{+}H_{s}.
\end{align*}
Finally, we define the continuous part $H^{c}$ of $H$ by $H_{t}^{c}\coloneqq H_{t}-H_{t}^{d}-H_{t}^{d,+}$.
Moreover, for each real-valued, $\mathbb{F}$-predictable process $H$, whose total variation is $[0,\infty]$-valued, $\mathbb{F}$-predictable and $P$-a.s.~finite, we define the processes $H^{\uparrow}$ and $H^{\downarrow}$ via
\begin{align*}
    H_{t}^{\uparrow}(\omega)&\coloneqq\frac{\vert H(\omega)\vert_{t}+H_{t}(\omega)}{2}\mathbbm{1}_{\llbracket 0,\sigma_{1}\wedge\sigma_{2}\llbracket}(\omega,t)\mathbbm{1}_{\llbracket 0,\sigma_{3}\rrbracket}(\omega,t),\\
    H_{t}^{\downarrow}(\omega)&\coloneqq\frac{\vert H(\omega)\vert_{t}-H_{t}(\omega)}{2}\mathbbm{1}_{\llbracket 0,\sigma_{1}\wedge\sigma_{2}\llbracket}(\omega,t)\mathbbm{1}_{\llbracket 0,\sigma_{3}\rrbracket}(\omega,t),
\end{align*}
where $\sigma_{1}$, $\sigma_{2}$ and $\sigma_{3}$ are $\mathbb{F}$-stopping times defined as
\begin{align*}
    \sigma_{1}&\coloneqq\inf\{t>0\colon\vert H\vert_{t-}=\infty\},\>\sigma_{2}\coloneqq\inf\{t>0\colon\vert \Delta H_{t}\vert=\infty\},\>\sigma_{3}\coloneqq\inf\{t>0\colon\vert\Delta_{+}H_{t}\vert=\infty\}.
\end{align*}
Since $\sigma_{1}$ and $\sigma_{2}$ are predictable stopping times (see Remark~IV.87(d) in \cite{dellacherie1978} together with Remark~E of the preliminary section ``Complements to Chapter~IV'' in \cite{dellacherie1982}), the processes $H^{\uparrow}$ and $H^{\downarrow}$ are $\mathbb{F}$-predictable. Furthermore, $H^{\uparrow}$ and $H^{\downarrow}$ are $P$-a.s.~increasing and satisfy $P[H_{T}^{\uparrow}<\infty]=P[H_{T}^{\downarrow}<\infty]=1$ because $P[\sigma_{1}=\infty]=P[\sigma_{2}=\infty]=P[\sigma_{3}=\infty]=1$. We then have that $H_{t}(\omega)=H_{t}^{\uparrow}(\omega)-H_{t}^{\downarrow}(\omega)$ for almost every $\omega$ for every $t\in[0,T]$. In the sequel, we will refer to this representation as the \emph{canonical decomposition}, and very often use the simplified notation $H=H^{\uparrow}-H^{\downarrow}$. In particular, in the spirit of \citet{campi2006}, we may say that the process $H$ is of finite variation as almost all of its paths have finite variation.

For a fixed model $\theta\in\Theta$, a strategy is called \emph{self-financing under transaction costs $\lambda$}, if
\begin{equation}\label{eq:selffin}
    \int_{s}^{t}\diff H_{u}^{0}\leq-\int_{s}^{t}S_{u}^{\theta}\diff H_{u}^{1,\uparrow}+\int_{s}^{t}(1-\lambda)S_{u}^{\theta}\diff H_{u}^{1,\downarrow}
\end{equation}
for all $0\leq s<t\leq T$, where the integrals
\begin{align}
    \int^t_s S_u dH^{1,\uparrow}_u&:= \int^t_s S_u dH_u^{1,\uparrow,c} + \sum_{s < u \leq t}S_{u-}\Delta H_u^{1,\uparrow} + \sum_{s \leq u < t} S_u \Delta_+ H_u^{1,\uparrow},\label{eq:stoch_int_up}\\
    \int^t_s(1-\lambda) S_u dH_u^{1,\downarrow}&:= \int^t_s (1-\lambda) S_u dH_u^{1, \downarrow,c} + \sum_{s < u \leq t} (1-\lambda)S_{u-}\Delta H^{1, \downarrow}_u + \sum_{s \leq u < t} (1-\lambda)S_u \Delta_+ H^{1,\downarrow}_u\label{eq:stoch_int_down}
\end{align}
can be defined pathwise by using Riemann-Stieltjes integrals on $\llbracket 0,\sigma_1\wedge\sigma_2\llbracket\cap\llbracket 0,\sigma_{3}\rrbracket\cap[0,T]$. For details on the above integrals \eqref{eq:stoch_int_up} and \eqref{eq:stoch_int_down}, we refer to Section 7 in \cite{czichowsky2016b}.
The self-financing condition \eqref{eq:selffin} then states that purchases and sales of the risky asset are accounted for in the riskless position:
\begin{align}
    \diff H_{t}^{0,c}&\leq-S_{t}^{\theta}\diff H_{t}^{1,\uparrow,c}+(1-\lambda)S_{t}^{\theta}\diff H_{t}^{1,\downarrow,c},\qquad 0\leq t\leq T,\label{eq:selffin_cont}\\
    \Delta H_{t}^{0}&\leq-S_{t-}^{\theta}\Delta H_{t}^{1,\uparrow}+(1-\lambda)S_{t-}^{\theta}\Delta H_{t}^{1,\downarrow},\qquad 0\leq t\leq T,\label{eq:selffin_disc-}\\
    \Delta_{+} H_{t}^{0}&\leq-S_{t}^{\theta}\Delta_{+} H_{t}^{1,\uparrow}+(1-\lambda)S_{t}^{\theta}\Delta_{+} H_{t}^{1,\downarrow},\qquad 0\leq t\leq T.\label{eq:selffin_disc+}
\end{align}
More precisely, we require that the processes $H^{0}=H^{0,\uparrow}-H^{0,\downarrow}$ and $H^{1}=H^{1,\uparrow}-H^{1,\downarrow}$ have increments that satisfy $\diff H_{t}^{0,\uparrow}\leq(1-\lambda)S_{t}^{\theta}\diff H_{t}^{1,\downarrow}$ and $\diff H_{t}^{0,\downarrow}\geq S_{t}^{\theta}\diff H_{t}^{1,\uparrow}$, which eventually leads to the self-financing condition \eqref{eq:selffin}, or in differential form to equations \eqref{eq:selffin_cont}, \eqref{eq:selffin_disc-} and \eqref{eq:selffin_disc+}. It is worth noting that the self-financing condition \eqref{eq:selffin}, and therefore \eqref{eq:selffin_cont}, \eqref{eq:selffin_disc-} and \eqref{eq:selffin_disc+}, is only well-defined on the set $\{\vert H^{0}\vert_{T}<\infty\}\cap\{\vert H^{1}\vert_{T}<\infty\}$. However, by assumption it holds that $P[\vert H^{0}\vert_{T}<\infty]=P[\vert H^{1}\vert_{T}<\infty]=1$, which implies that the self-financing condition is well-defined $P$-a.s.

For a fixed model $\theta\in\Theta$, a self-financing strategy $H$ is \emph{admissible} under transaction costs $\lambda$, if its \emph{liquidation value} $V^{\mathrm{liq}}(\theta,H)$ satisfies
\begin{equation}\label{eq:admissible}
    V_{t}^{\mathrm{liq}}(\theta,H)\coloneqq H_{t}^{0}+(H_{t}^{1})^{+}(1-\lambda)S_{t}^{\theta}-(H_{t}^{1})^{-}S_{t}^{\theta}\geq 0,\quad\text{a.s.},
\end{equation}
for all $t\in[0,T]$. For $x>0$ and a fixed model $\theta\in\Theta$, we denote by $\mathcal{H}^{\theta}(x)$ the set of all self-financing, admissible trading strategies under transaction costs $\lambda$, starting with initial endowment $(H_{0}^{0},H_{0}^{1})=(x,0)$. As for the self-financing condition, the admissibility condition \eqref{eq:admissible} is well-defined on $\llbracket 0,\sigma_1\wedge\sigma_2\llbracket\cap\llbracket 0,\sigma_{3}\rrbracket\cap[0,T]$.

In order to get towards a model-independent setup (that is, we want to consider self-financing trading strategies that are admissible for all models $\theta\in\Theta$), we pass to a dominating pair $(H^{0},H^{1})$ for each trading strategy $H\in\mathcal{H}^{\theta}(x)$ where equality holds true in \eqref{eq:selffin}. This way, we only have to specify one of the holdings, e.g., the number of stocks $H^{1}$. For a fixed model $\theta\in\Theta$ and $x>0$, we thus define the set
\begin{align*}
    \mathcal{A}^{\theta}(x)\coloneqq\big\{H^{1}\colon (H^{0},H^{1})\in\mathcal{H}^{\theta}(x),\>\diff H_{t}^{0}=-S_{t}^{\theta}\diff H_{t}^{1,\uparrow}+(1-\lambda)S_{t}^{\theta}\diff H_{t}^{1,\downarrow}\big\}.
\end{align*}
This definition is in line with the set of admissible trading strategies considered in the case of one single model, that is, $\Theta=\{\theta\}$, as discussed in \cite{czichowsky2016}. We will also refer to this situation as the \emph{non-robust} case. Moreover, by letting $\mathcal{A}(x)\coloneqq\bigcap_{\theta\in\Theta}\mathcal{A}^{\theta}(x)$, we obtain the analogue of the set of model-independent admissible trading strategies given in \cite{chau2019}.

Note that $H^{1}\in\mathcal{A}(x)$ does no longer depend on $\theta$. However, the holdings in the bond $H^{0}$ still depend on $\theta$. We will use the notation $H^{0,\theta}$ to indicate this dependence. In particular, we define $H_{t}^{0,\theta}\coloneqq x+H_{t}^{0,\theta,\uparrow}-H_{t}^{0,\theta,\downarrow}$ with
\begin{align}\label{eq:H0_theta}
    H_{t}^{0,\theta,\downarrow}\coloneqq\int_{0}^{t}S_{u}^{\theta}\diff H_{u}^{1,\uparrow},\quad\text{and}\quad H_{t}^{0,\theta,\uparrow}\coloneqq\int_{0}^{t}(1-\lambda)S_{u}^{\theta}\diff H_{u}^{1,\downarrow}.
\end{align}
Moreover, we write $V_{t}^{\mathrm{liq}}(\theta,H^{1})$ to indicate that $H^{0}$ in \eqref{eq:admissible} is defined via \eqref{eq:H0_theta}. We also notice that the mapping $H^{1}\mapsto V_{t}^{\mathrm{liq}}(\theta,H^{1})$ is concave. To see this statement, we first note that $x^{+}=\max(x,0)$ and $x^{-}=\max(-x,0)$ are convex functions. By definition (see, e.g., equation (5.1) in Section~X.5 in \cite{doob1994}) we have
\begin{align*}
    H_{t}^{1,\downarrow}(\omega)&=\sup\bigg\{\sum_{t_{k}\in\Pi}(H_{t_{k}}^{1}(\omega)-H_{t_{k-1}}^{1}(\omega))^{-}\mathbbm{1}_{\llbracket 0,\sigma_{1}\wedge\sigma_{2}\llbracket\cap\llbracket 0,\sigma_{3}\rrbracket}(\omega,t)\bigg|\Pi\text{ is a partition of }[0,t]\bigg\},
\end{align*}
which implies that $(\alpha H_{t}^{1}+(1-\alpha)\widetilde{H}_{t}^{1})^{\downarrow}\leq\alpha H_{t}^{1,\downarrow}+(1-\alpha)\widetilde{H}_{t}^{1,\downarrow}$ for all $\alpha\in[0,1]$. Now, observe that
\begin{align*}
    V_{t}^{\mathrm{liq}}(\theta,H^{1})=-\lambda\int_{0}^{t}S_{u}^{\theta}\diff H_{u}^{1,\downarrow}-\int_{0}^{t}S_{u}^{\theta}\diff H_{u}^{1}+H_{t}^{1}S_{t}^{\theta}-(H_{t}^{1})^{+}\lambda S_{t}^{\theta},
\end{align*}
which, together with the fact that $S^{\theta}$ is strictly positive, implies that  $V_{t}^{\mathrm{liq}}(\theta,H^{1})$ is concave in $H^{1}$.
\ 

Now, we have everything in place to formulate the optimisation problem. For this reason, we consider an investor whose preferences are modelled by a standard utility function\footnote{That is a strictly concave, non-decreasing and continuously differentiable function satisfying the Inada conditions $U'(0)=\lim_{x\to 0}U'(x)=\infty$ and $U'(\infty)=\lim_{x\to\infty}U'(x)=0$.} $U\colon(0,\infty)\to\mathbb{R}$. For a given initial capital $x>0$, the investor wants to maximise the expected utility of terminal wealth with respect to the worst-case scenario of all possible models. This means that the investor wants to find the optimal strategy $\widehat{H}^{1}\in\mathcal{A}(x)$ that maximises $\inf_{\theta\in\Theta}\E\big[U\big(V_{T}^{\mathrm{liq}}\big(\theta,H^{1}\big)\big)\big]$. The value function of this \emph{primal optimisation problem} is denoted by
\begin{equation}\label{eq:primalproblem}
    u(x)\coloneqq\sup_{H^{1}\in\mathcal{A}(x)}\inf_{\theta\in\Theta}\E\big[U\big(V_{T}^{\mathrm{liq}}\big(\theta,H^{1}\big)\big)\big].
\end{equation}
In the sequel, we answer the question of whether, and under which assumptions, the robust primal problem \eqref{eq:primalproblem} admits a solution.

\section{Main result}\label{sec:main_result}
In the frictionless case, the Fundamental Theorem of Asset Pricing states that the no-arbitrage condition is equivalent to the property that the price process admits an equivalent local martingale measure (see \citet{delbaen1994}). In the setting of transaction costs, the notion of \emph{consistent price systems} plays a role similar to the notion of equivalent martingale measures in the frictionless case.
\begin{definition}\label{def:Consistent_price_systems}
    Fix $0<\lambda<1$ and $\theta\in\Theta$. The strictly positive adapted c{\`a}dl{\`a}g process $S^{\theta}$ satisfies the condition $(\mathrm{CPS}^{\lambda})$ of having a \emph{$\lambda$-consistent price system}, if there exists a pair of processes $Z^{\theta}=(Z_{t}^{0,\theta},Z_{t}^{1,\theta})_{0\leq t\leq T}$, consisting of a density process $Z^{0,\theta}=(Z_{t}^{0,\theta})_{0\leq t\leq T}$ of an equivalent local martingale measure $Q^{\theta}\approx P$ for a price process $\widetilde{S}^{\theta}=(\widetilde{S}^{\theta}_{t})_{0\leq t\leq T}$ evolving in the bid-ask spread $[(1-\lambda)S^{\theta},S^{\theta}]$, and $Z^{1,\theta}=Z^{0,\theta}\widetilde{S}^{\theta}$. In particular, $\widetilde{S}^{\theta}$ satisfies
    \begin{align}\label{eq:CPS}
        (1-\lambda)S_{t}^{\theta}\leq\widetilde{S}_{t}^{\theta}\leq S_{t}^{\theta},\qquad 0\leq t\leq T.
    \end{align}
    We further say that $S^{\theta}$ satisfies the condition $(\mathrm{CPS}^{\lambda})$ \emph{locally} of having a local $\lambda$-consistent price system, if there exists a strictly positive stochastic process $Z^{\theta}=(Z^{0,\theta},Z^{1,\theta})$ and a localising sequence $(\tau_{n})_{n\geq 0}$ of stopping times, such that $(Z^{\theta})^{\tau_{n}}$ is a $\lambda$-consistent prices system for the stopped process $(\widetilde{S}^{\theta})^{\tau_{n}}$ for each $n\geq 0$. We denote the space of all such processes by $\mathcal{Z}^{\theta}$ and $\mathcal{Z}_{\mathrm{loc}}^{\theta}$, respectively.
\end{definition}
We impose the existence of local consistent price systems for every model $\theta\in\Theta$.
\begin{assumption}\label{assump:CPS_all}
    For each $\theta\in\Theta$ and for some $0<\lambda'<\lambda$, the price process $S^{\theta}$ satisfies $(\mathrm{CPS}^{\lambda'})$ locally.
\end{assumption}
In the non-robust setting, i.e., for a fixed model $\theta\in\Theta$ and $x>0$, we define
\begin{align}\label{def:C}
    \mathcal{C}^{\theta}(x)\coloneqq\big\{g\in L_{+}^{0}(\Omega,\mathcal{F},P)\colon g\leq V_{T}^{\mathrm{liq}}(\theta,H),\text{ for some }H\in\mathcal{H}^{\theta}(x)\big\}.
\end{align}
This is the set of terminal positions $g$ that one can superreplicate with an admissible trading strategy $H$ and initial endowment $x$. Note that $\mathcal{C}^{\theta}(x)=x\mathcal{C}^{\theta}(1)$ for all $x>0$. Since we are not interested in an analysis of the dual problem on the level of stochastic processes, we can define the dual variables simply on the level of random variables by setting
\begin{align}\label{def:D}
    \mathcal{D}^{\theta}(y)\coloneqq\big\{h\in L_{+}^{0}(\Omega,\mathcal{F},P)\colon\E[gh]\leq y\text{ for all $g\in\mathcal{C}^{\theta}(1)$}\big\}.
\end{align}
Note that $\mathcal{D}^{\theta}(y)=y\mathcal{D}^{\theta}(1)$ $\forall y>0$, and by Proposition~\ref{prop:strong_supermartingale} $\{yZ^0_T\colon(Z^0,Z^1)\in\mathcal{Z}_{\mathrm{loc}}^{\theta}\}\subseteq\mathcal{D}^{\theta}(y)$ so that $\mathcal{D}^{\theta}(y)\ne\emptyset$.
Moreover, by definition the set $\mathcal{D}^{\theta}(1)$ corresponds to the polar $(\mathcal{C}^{\theta}(1))^{\circ}$ of $\mathcal{C}^{\theta}(1)$ in $L_{+}^{0}(P)$ as in Definition~1.2 of \citet{brannath1999}. We verify in Lemma~\ref{lem:C(x)_convex_solid_closed} that the sets $\mathcal{C}^{\theta}(1)$ and $\mathcal{D}^{\theta}(1)$ satisfy the properties of the sets $\mathcal{C}$ and $\mathcal{D}$ in Proposition~3.1 of \citet{kramkov1999}. We can therefore use the abstract version of the duality result for utility maximisation of random variables in Theorem~3.1 in \cite{kramkov1999} in each single model $S^{\theta}$.

Using the sets $\mathcal{C}^{\theta}(x)$ and $\mathcal{D}^{\theta}(y)$ allows us to define the single model value functions. In particular, the primal and dual value functions for the $\theta$-model are given by
\begin{align*}
    u^{\theta}(x)\coloneqq\sup_{f\in\mathcal{C}^{\theta}(x)}\E[U(f)],\quad\text{and}\quad j^{\theta}(y)\coloneqq\inf_{h\in\mathcal{D}^{\theta}(y)}\E[J(h)],
\end{align*}
where $J(y)\coloneqq\sup_{x>0}(U(x)-xy)$, $y>0$, is the Legendre transform of $-U(-x)$. For our purpose, we need the following assumption.
\begin{assumption}\label{assump:AEU_primalrpoblem}
    The asymptotic elasticity of $U$ is strictly less than one, that is,
    \begin{align*}
        \mathrm{AE}(U)\coloneqq\limsup_{x\to\infty}\frac{xU'(x)}{U(x)}<1,
    \end{align*}
    and for each model $\theta\in\Theta$, the primal value function $u^{\theta}(x)$ is finite for \emph{some} $x>0$ and hence \emph{all} $x>0$ by concavity of $u^{\theta}(x)$.
\end{assumption}
\begin{remark}
In their paper \citet{chau2019}, instead of Assumption \ref{assump:AEU_primalrpoblem}, use the assumption that $j^{\theta}(y)$, $y>0$, is finite for all models $\theta\in\Theta$ (see Assumption~3.5 in \cite{chau2019}). However, note that Assumption~\ref{assump:CPS_all} and Assumption~\ref{assump:AEU_primalrpoblem} imply Assumption~3.5 in \cite{chau2019} (see Theorem~2 and the subsequent Note~2 in \cite{kramkov2003} together with Theorem~2.2 of \cite{kramkov1999}). In particular, to make use of Note~2 in \cite{kramkov2003} we need the existence of some $y_{0}>0$ such that $j^{\theta}(y)<\infty$ for all $y>y_{0}$. This existence is guaranteed by Theorem~3.1 of \cite{kramkov1999}, which is applicable in our setup because of Lemma~\ref{lem:C(x)_convex_solid_closed} in the Appendix. Moreover, Assumption~\ref{assump:CPS_all} is standard and Assumption~\ref{assump:AEU_primalrpoblem} is satisfied for most popular utility functions, like logarithmic and power utility.
\end{remark}
\begin{remark}\label{rem:equiv_assump_dualproblem}
    Under Assumptions~\ref{assump:CPS_all} and \ref{assump:AEU_primalrpoblem}, we obtain by Theorem~3.2 in \cite{czichowsky2016} that
    \begin{align*}
        (u^{\theta})'(\infty)=\lim_{x\to\infty}(u^{\theta})'(x)=0,\quad\forall\theta\in\Theta,
    \end{align*}
    which may be restated as
    \begin{align*}
        \lim_{x\to\infty}\frac{u^{\theta}(x)}{x}=0,\quad\forall\theta\in\Theta.
    \end{align*}
    As explained below Theorem 3.2 of \citet{czichowsky2016}, the condition that $S^{\theta}=(S_{t}^{\theta})_{0\leq t\leq T}$ satisfies $(\mathrm{CPS}^{\lambda'})$ locally for \emph{all} $0<\lambda'<\lambda$ is only necessary to argue that there is no ``duality gap'' in part (4) of Theorem 3.2 of \citet{czichowsky2016}. Since we do not require this conclusion for our results, we can use the weaker assumption that $S^{\theta}=(S_{t}^{\theta})_{0\leq t\leq T}$ satisfies $(\mathrm{CPS}^{\lambda'})$ locally for \emph{some} $0<\lambda'<\lambda$.
\end{remark}
The following theorem is the main result of this paper. It extends Theorem~3.6 in \cite{chau2019} for utility functions defined on the positive half-line from continuous price processes to general price processes with jumps.
\begin{thm}\label{thm:robust_primal_problem}
    Let $x>0$. Under Assumptions \ref{assump:CPS_all} and \ref{assump:AEU_primalrpoblem}, the robust utility maximisation problem \eqref{eq:primalproblem} admits a solution, i.e., there is $\widehat{H}^{1}\in\mathcal{A}(x)$ satisfying
    \begin{align*}
        u(x)=\inf_{\theta\in\Theta}\E\big[U\big(V_{T}^{\mathrm{liq}}\big(\theta,\widehat{H}^{1}\big)\big)\big].
    \end{align*}
    When $U$ is bounded from above, the same conclusion holds assuming only that there exists (at least) one $\theta'\in\Theta$ for which $(\mathrm{CPS}^{\lambda'})$ locally holds true for some $\lambda'\in(0,\lambda)$.
\end{thm}

\section{Proof of the main theorem}\label{sec:proof_main_result}
To prove Theorem~\ref{thm:robust_primal_problem}, we need the following four results. The first result states that for a fixed model $\theta$, the value process with respect to a consistent price system $(\widetilde{S}^{\theta},Q^{\theta})$ is an optional strong supermartingale under $Q^{\theta}$. Optional strong supermartingales have been introduced by \citet{mertens1972} as a generalisation of the notion of a c{\`a}dl{\`a}g supermartingale. We recall the definition (see Definition~1 of Appendix~I in \cite{dellacherie1982}).
\begin{definition}\label{def:optional_strong_supermartingale}
    An optional process $X=(X_{t})_{t\geq 0}$ is an \emph{optional strong supermartingale}, if:
    \begin{enumerate}
        \item For every bounded stopping time $\tau$, $X_{\tau}$ is integrable.\label{itm:OSS_integrable}
        \item For every pair of bounded stopping times $\sigma$ and $\tau$, such that $\sigma\leq\tau$, we have
        \begin{align*}
            \E[X_{\tau}|\cF_{\sigma}]\leq X_{\sigma},\quad P\text{-a.s.}
        \end{align*}\label{itm:OSS_supermartingale_property}
    \end{enumerate}
\end{definition}
For further discussion of optional strong supermartingales, we refer to Appendix~I in \citet{dellacherie1982}.
\begin{prop}\label{prop:strong_supermartingale}
    For a fixed model $\theta\in\Theta$ with price process $S^{\theta}=(S_{t}^{\theta})_{0\leq t\leq T}$, transaction costs $0<\lambda<1$, and an admissible self-financing trading strategy $H\in\mathcal{H}^{\theta}(x)$, $x>0$, as above, suppose that $(Z^{0,\theta},Z^{0,\theta}\widetilde{S}^{\theta})\in\mathcal{Z}_{\mathrm{loc}}^{\theta}$ is a locally consistent price system under transaction costs $\lambda$. Then, the process $\widetilde{V}(\theta,H)=(\widetilde{V}_{t}(\theta,H))_{0\leq t\leq T}$, defined by
    \begin{align*}
        \widetilde{V}_{t}(\theta,H)\coloneqq H_{t}^{0}+H_{t}^{1}\widetilde{S}_{t}^{\theta},\qquad 0\leq t\leq T,
    \end{align*}
    satisfies $\widetilde{V}_{t}(\theta,H)\geq V_{t}^{\mathrm{liq}}(\theta,H)$ almost surely, and $Z^{0,\theta}(H^{0}+H^{1}\widetilde{S}^{\theta})$ is an optional strong supermartingale.
\end{prop}
\begin{proof}
    In the setting where the filtration $\mathbb{F}$ satisfies the usual conditions, the proof of Proposition~\ref{prop:strong_supermartingale} is given in \cite{schachermayer2014b} (see Proposition 2). In the current setup, where $\mathbb{F}$ is right-continuous but not complete, we may still follow the lines of the proof given in \cite{schachermayer2014b}, but with special care on two technicalities.
    
    First of all, the main idea of the proof is to show that $\widetilde{V}(\theta,H)$ can be decomposed into the difference of a local martingale and an increasing $\mathbb{F}$-predictable process. The existence of this so-called Mertens decomposition ensures that $Z^{0,\theta}(H^{0}+H^{1}\widetilde{S}^{\theta})$ is an optional strong supermartingale (see Theorem~20 of Appendix~I in \cite{dellacherie1982}). This result not only holds under the usual conditions, but also under the weaker assumption that $\mathbb{F}$ is incomplete (see Remark~21 of Appendix~I in \cite{dellacherie1982}).
    
    Second of all, to eventually replicate the proof of Proposition~2 in \cite{schachermayer2014b}, we need to take special care when it comes to exhausting the jumps of $H=(H^{0},H^{1})$ by a countable sequence of disjoint graphs of stopping times. Since $\mathbb{F}$ is incomplete, we can not make use of Theorem~IV.117 in \cite{dellacherie1982}, since this result requires the usual conditions. However, we may instead use Theorem~D together with Remark~E of the preliminary section ``Complements to Chapter~IV'' in \cite{dellacherie1982} to exhaust the jump times of $H$.
\end{proof}
As a second result, we also need the following bipolar relation for the non-robust models.
\begin{lemma}\label{lem:bipolar_relation_model_based}
    Fix $x,y>0$. Suppose that $S^{\theta}$ satisfies $(\mathrm{CPS}^{\lambda'})$ locally for some $\lambda'\in(0,\lambda)$. Then, a random variable $g\in L_{+}^{0}(P)$ satisfies $g\in\mathcal{C}^{\theta}(x)$ if and only if $\E[gh]\leq xy$ for all $h\in\mathcal{D}^{\theta}(y)$.
\end{lemma}
\begin{proof}
    By assertion~\eqref{itm:bipolarity} of Lemma~\ref{lem:C(x)_convex_solid_closed} we know that $g\in\mathcal{C}^{\theta}(1)$ if and only if $\E[gh]\leq 1$ for all $h\in\mathcal{D}^{\theta}(1)$. This result is already enough to conclude the proof, since for all $x,y>0$ it holds that $\mathcal{C}^{\theta}(x)=x\mathcal{C}^{\theta}(1)$ and $\mathcal{D}^{\theta}(y)=y\mathcal{D}^{\theta}(1)$.
\end{proof}
The third result we need, extends Proposition 3.4 in \cite{campi2006} to a more general, model-independent view. In particular, while the result from Proposition 3.4 in \cite{campi2006} holds for predictable processes, our result in Proposition~\ref{prop:cc_general} is true for processes that are $\mathbb{F}$-predictable and whose total variation is $\mathbb{F}$-predictable and $P$-a.s.~finite. The need for such a result stems from the more general assumption that the filtration $\mathbb{F}$ is not complete in our setup.
\begin{prop}\label{prop:cc_general}
    Let $(H^{1,n})_{n\in\mathbb{N}}\subseteq\mathcal{A}(x)$ for some $x>0$ with canonical decompositions $H^{1,n}=H^{1,n,\uparrow}-H^{1,n,\downarrow}$. Assume that there is $\theta\in\Theta$ so that $S^{\theta}$ satisfies $(\mathrm{CPS}^{\lambda'})$ locally, for some $0<\lambda'<\lambda$. Then there exist processes $H^{1,\uparrow}$ and $H^{1,\downarrow}$, both $[0,\infty]$-valued, increasing, $\mathbb{F}$-predictable, and $P$-a.s.~finite,
    and a sequence of convex combinations
    \begin{align*}
        (\widetilde{H}^{1,n,\uparrow},\widetilde{H}^{1,n,\downarrow})_{n\in\mathbb{N}}\subseteq\mathrm{conv}((H^{1,n,\uparrow},H^{1,n,\downarrow}),(H^{1,n+1,\uparrow},H^{1,n+1,\downarrow}),\dots),
    \end{align*}
    such that $(\widetilde{H}^{1,n,\uparrow},\widetilde{H}^{1,n,\downarrow})$ converges for almost every $\omega$ for every $t\in[0,T]$ to $(H^{1,\uparrow},H^{1,\downarrow})$, i.e.,
    \begin{align*}
        P\big[(\widetilde{H}_t^{1,n,\uparrow},\widetilde{H}_t^{1,n,\downarrow})\to (H_t^{1,\uparrow},H_t^{1,\downarrow}),\quad\forall t\in[0,T]\big]=1.
    \end{align*}
    In particular, the process $H^{1}=H^{1,\uparrow}\mathbbm{1}_{\llbracket 0,\sigma_{1}\wedge\sigma_{2}\llbracket}\mathbbm{1}_{\llbracket 0,\sigma_{3}\rrbracket}-H^{1,\downarrow}\mathbbm{1}_{\llbracket 0,\sigma_{1}\wedge\sigma_{2}\llbracket}\mathbbm{1}_{\llbracket 0,\sigma_{3}\rrbracket}$, with stopping times $\sigma_{1}\coloneqq\inf\{t>0\colon H_{t-}^{1,\uparrow}=\infty\text{ or }H_{t-}^{1,\downarrow}=\infty\}$, $\sigma_{2}\coloneqq\inf\{t>0\colon \Delta H_{t}^{1,\uparrow}=\infty\text{ or }\Delta H_{t}^{1,\downarrow}=\infty\}$, and $\sigma_{3}\coloneqq\inf\{t>0\colon \Delta_{+}H_{t}^{1,\uparrow}=\infty\text{ or }\Delta_{+}H_{t}^{1,\downarrow}=\infty\}$, satisfies $H^{1}\in\mathcal{A}(x)$, and the sequence $(\widetilde{H}^{1,n})_{n\in\mathbb{N}}\subseteq\mathrm{conv}(H^{1,n},H^{1,n+1},\dots)$ converges to $H^{1}$ for almost every $\omega$ for every $t\in[0,T]$.
\end{prop}
In the proof of  Proposition~\ref{prop:cc_general}, which is given in Appendix~\ref{app:proof_con_compct}, we will also need the following result. It shows that the pointwise convergence of integrators of finite variation is sufficient for the convergence of the integrals of a fixed c{\`a}dl{\`a}g function. While this result is not true for the standard Riemann-Stieltjes integrals, it is true for our notion of the integral that is motivated by self-financing trading in financial markets (see equations \eqref{eq:stoch_int_up} and \eqref{eq:stoch_int_down}). The reason is that, for trades immediately before a predictable stopping time, the price paid is the left limit of the price process. The proof of Lemma~\ref{lem:conint} is also given in Appendix~\ref{app:proof_con_compct}.
\begin{lemma}\label{lem:conint}
    Let $x>0$ and consider the sequence $(H^{1,n})_{n\in\mathbb{N}}\subseteq\mathcal{A}(x)$ with canonical decomposition $H^{1,n}=H^{1,n,\uparrow}-H^{1,n,\downarrow}$. Moreover, assume that $H^{1,\uparrow}$ and $H^{1,\downarrow}$ are $[0,\infty]$-valued, increasing, $\mathbb{F}$-predictable processes, which satisfy $P[H_{T}^{1,\uparrow}<\infty]=P[H_{T}^{1,\downarrow}<\infty]=1$,
    such that
    \begin{align}\label{eq:conv_H1_decomp}
        P\big[H_{t}^{1,n,\uparrow}\to H_{t}^{1,\uparrow},\ \forall t\in[0,T]\big]=1,\quad\text{and}\quad P\big[H_{t}^{1,n,\downarrow}\to H_{t}^{1,\downarrow},\ \forall t\in[0,T]\big]=1.
    \end{align}
    Then, for all $\theta\in\Theta$, the sequences $(H^{0,\theta,n,\downarrow})_{n\in\mathbb{N}}$ and $(H^{0,\theta,n,\uparrow})_{n\in\mathbb{N}}$
    defined via \eqref{eq:H0_theta}, i.e.,
    \begin{align}\label{eq:H_0_n_down_and_H_0_n_up}
        H_{t}^{0,\theta,n,\downarrow}=\int_{0}^{t}S_{u}^{\theta}\diff H_{u}^{1,n,\uparrow},\quad\text{and}\quad H_{t}^{0,\theta,n,\uparrow}=\int_{0}^{t}(1-\lambda)S_{u}^{\theta}\diff H_{u}^{1,n,\downarrow},
    \end{align}
    converge for almost every $\omega$ for every $t\in[0,T]$ to the processes $H_{t}^{0,\theta,\downarrow}=\int_{0}^{t}S_{u}^{\theta}\diff H_{u}^{1,\uparrow}$ and $H_{t}^{0,\theta,\uparrow}=\int_{0}^{t}(1-\lambda)S_{u}^{\theta}\diff H_{u}^{1,\downarrow}$, respectively, i.e.,
    \begin{align}\label{eq:conv_H0_decomp}
        P\big[H_{t}^{0,\theta,n,\uparrow}\to H_{t}^{0,\theta,\uparrow},\ \forall t\in[0,T]\big]=1,\quad\text{and}\quad P\big[H_{t}^{0,\theta,n,\downarrow}\to H_{t}^{0,\theta,\downarrow},\ \forall t\in[0,T]\big]=1.
    \end{align}
    In particular, $H^{0,\theta,\uparrow}$ and $H^{0,\theta,\downarrow}$ are $[0,\infty]$-valued, increasing and $\mathbb{F}$-predictable processes satisfying $P[H_{T}^{0,\theta,\uparrow}<\infty]=P[H_{T}^{0,\theta,\downarrow}<\infty]=1$, and therefore the process
    \begin{align*}
        H^{1}=H^{1,\uparrow}\mathbbm{1}_{\llbracket 0,\sigma_{1}\wedge\sigma_{2}\llbracket}\mathbbm{1}_{\llbracket 0,\sigma_{3}\rrbracket}-H^{1,\downarrow}\mathbbm{1}_{\llbracket 0,\sigma_{1}\wedge\sigma_{2}\llbracket}\mathbbm{1}_{\llbracket 0,\sigma_{3}\rrbracket},
    \end{align*}
    $\sigma_{1}\coloneqq\inf\{t>0\colon H_{t-}^{1,\uparrow}=\infty\text{ or }H_{t-}^{1,\downarrow}=\infty\}$, $\sigma_{2}\coloneqq\inf\{t>0\colon \Delta H_{t}^{1,\uparrow}=\infty\text{ or }\Delta H_{t}^{1,\downarrow}=\infty\}$, $\sigma_{3}\coloneqq\inf\{t>0\colon \Delta_{+}H_{t}^{1,\uparrow}=\infty\text{ or }\Delta_{+}H_{t}^{1,\downarrow}=\infty\}$, 
    satisfies $H^{1}\in\mathcal{A}(x)$.
\end{lemma}
We are now ready to prove our main result.
\begin{proof}[Proof of Theorem \ref{thm:robust_primal_problem}]
    Let $(H^{1,n})_{n\in\mathbb{N}}\subseteq\mathcal{A}(x)$ be a maximising sequence, i.e.,
    \begin{align*}
        \inf_{\theta\in\Theta}\E\big[U\big(V_{T}^{\mathrm{liq}}\big(\theta,H^{1,n}\big)\big)\big]\nearrow u(x),\quad\text{as }n\to\infty.
    \end{align*}
    By Proposition~\ref{prop:cc_general} there is a sequence $(\widetilde{H}^{1,n})_{n\in\mathbb{N}}\subseteq\mathrm{conv}(H^{1,n},H^{1,n+1},\dots)$ and $\widehat{H}^{1}\in\mathcal{A}(x)$, such that, for almost every $\omega\in\Omega$, we have that $\widetilde{H}^{1,n}_t$ converges for every $t\in[0,T]$ to $\widehat{H}^{1}_t$. Since the utility function is concave and increasing, and because $H^{1}\mapsto V_{T}^{\mathrm{liq}}(\theta,H^{1})$ is concave, we obtain that $(\widetilde{H}^{1,n})_{n\in\mathbb{N}}$ is also a maximising sequence, because
    \begin{align*}
        \inf_{\theta\in\Theta}\E\big[U\big(V_{T}^{\mathrm{liq}}\big(\theta,\widetilde{H}^{1,n}\big)\big)\big]\geq\inf_{\theta\in\Theta}\E\big[U\big(V_{T}^{\mathrm{liq}}\big(\theta,H^{1,n}\big)\big)\big]\to u(x),\quad\text{as }n\to\infty.
    \end{align*}
    We claim that $\widehat{H}^{1}$ is an optimal solution to \eqref{eq:primalproblem}. For this purpose, we denote by $U^{+}$ and $U^{-}$ the positive and negative parts of the function $U$. Since  for almost every $\omega\in\Omega$, we have that $\widetilde{H}^{1,n}_t\to\widehat{H}^{1}_t$ for every $t\in[0,T]$, we get that $V_{T}^{\mathrm{liq}}(\theta,\widetilde{H}^{1,n})\to V_{T}^{\mathrm{liq}}(\theta,\widehat{H}^{1})$ almost surely for each $\theta\in\Theta$ by Lemma~\ref{lem:conint}. Hence, from Fatou's lemma, we deduce that
    \begin{align*}
        \liminf_{n\to\infty}\E\big[U^{-}\big(V_{T}^{\mathrm{liq}}\big(\theta,\widetilde{H}^{1,n}\big)\big)\big]\geq\E\big[U^{-}\big(V_{T}^{\mathrm{liq}}\big(\theta,\widehat{H}^{1}\big)\big)\big],
    \end{align*}
    for each model $\theta\in\Theta$. The optimality of $\widehat{H}$ will follow if we show that
    \begin{align}\label{eq:L1_boundedness_U+}
        \lim_{n\to\infty}\E\big[U^{+}\big(V_{T}^{\mathrm{liq}}\big(\theta,\widetilde{H}^{1,n}\big)\big)\big]=\E\big[U^{+}\big(V_{T}^{\mathrm{liq}}\big(\theta,\widehat{H}^{1}\big)\big)\big],
    \end{align}
    for each model $\theta\in\Theta$. If $U(\infty)\leq 0$, then there is nothing to prove. So we assume that $U(\infty)>0$.

    We claim that the sequence $(U^{+}(V_{T}^{\mathrm{liq}}(\theta,\widetilde{H}^{1,n})))_{n\in\mathbb{N}}$ is uniformly integrable for each $\theta\in\Theta$. This assertion is equivalent to the validity of \eqref{eq:L1_boundedness_U+}. Suppose, by contradiction, that the sequence is not uniformly integrable for some $\theta$. Then, using the characterization of uniform integrability given in \cite{doob1994}, Theorem VI.16, and passing if necessary to a subsequence still denoted by $(\widetilde{H}^{1,n})_{n\in\mathbb{N}}$, we can find a constant $\alpha>0$ and disjoint sets $A_{n}\in\mathcal{F}$, $n\in\mathbb{N}$, such that
    \begin{align*}
        \E\big[U^{+}\big(V_{T}^{\mathrm{liq}}\big(\theta,\widetilde{H}^{1,n}\big)\big)\mathbbm{1}_{A_{n}}\big]\geq\alpha\quad\text{for }n\geq 1.
    \end{align*}
    We define the sequence of random variables $(h^{n})_{n\in\mathbb{N}}$ by
    \begin{align*}
        h^{n}\coloneqq x_{0}+\sum_{k=1}^{n}V_{T}^{\mathrm{liq}}\big(\theta,\widetilde{H}^{1,k}\big)\mathbbm{1}_{A_{k}},
    \end{align*}
    where $x_{0}=\inf\big\{x>0\colon U(x)\geq 0\big\}$. It follows immediately that
    \begin{align*}
        \E[U(h^{n})]\geq\sum_{k=1}^{n}\E\big[U^{+}\big(V_{T}^{\mathrm{liq}}\big(\theta,\widetilde{H}^{1,k}\big)\big)\mathbbm{1}_{A_{k}}\big]\geq n\alpha.
    \end{align*}
    On the other hand, for any $f\in\mathcal{D}^{\theta}(1)$ we have by Proposition \ref{prop:strong_supermartingale} that
    \begin{align*}
        \E[h^{n}f]\leq x_{0}+\sum_{k=1}^{n}\E\big[V_{T}^{\mathrm{liq}}\big(\theta,\widetilde{H}^{1,k}\big)f\big]\leq x_{0}+nx.
    \end{align*}
    Hence, by Lemma \ref{lem:bipolar_relation_model_based}, we obtain $h^{n}\in\mathcal{C}^{\theta}(x_{0}+nx)$. Therefore, we have
    \begin{align*}
        \limsup_{x\to\infty}\frac{u^{\theta}(x)}{x}\geq\limsup_{n\to\infty}\frac{\E[U(h^{n})]}{x_{0}+nx}\geq\limsup_{n\to\infty}\frac{n\alpha}{x_{0}+nx}=\frac{\alpha}{x}>0.
    \end{align*}
    By Remark \ref{rem:equiv_assump_dualproblem}, this strict inequality is a contradiction to our Assumptions \ref{assump:CPS_all} and \ref{assump:AEU_primalrpoblem}. As a result, $(U^{+}(V_{T}^{\mathrm{liq}}(\theta,\widetilde{H}^{1,n})))_{n\in\mathbb{N}}$ indeed is a uniformly integrable sequence for each $\theta\in\Theta$.

    As before, we use that $V_{T}^{\mathrm{liq}}(\theta,\widetilde{H}^{1,n})\to V_{T}^{\mathrm{liq}}(\theta,\widehat{H}^{1})$ almost surely for each $\theta\in\Theta$ by Lemma~\ref{lem:conint}. Therefore, Fatou's lemma and uniform integrability imply
    \begin{align*}
        \limsup_{n\to\infty}\Big(\inf_{\theta\in\Theta}\E\big[U\big(V_{T}^{\mathrm{liq}}\big(\theta,\widetilde{H}^{1,n}\big)\big)\big]\Big)&\leq\inf_{\theta\in\Theta}\Big(\limsup_{n\to\infty}\E\big[U\big(V_{T}^{\mathrm{liq}}\big(\theta,\widetilde{H}^{1,n}\big)\big)\big]\Big)\\[0.5em]
        &\leq\inf_{\theta\in\Theta}\E\big[U\big(V_{T}^{\mathrm{liq}}\big(\theta,\widehat{H}^{1}\big)\big)\big],
    \end{align*}
    which proves that $\widehat{H}^{1}$ is an optimal solution to \eqref{eq:primalproblem}.

    In the case where $U$ is bounded from above, we use Fatou's lemma to get
    \begin{align*}
        \limsup_{n\to\infty}\E\big[U^{+}\big(V_{T}^{\mathrm{liq}}\big(\theta,\widetilde{H}^{1,n}\big)\big)\big]\leq\E\big[U^{+}\big(V_{T}^{\mathrm{liq}}\big(\theta,\widehat{H}^{1}\big)\big)\big]
    \end{align*}
    instead of \eqref{eq:L1_boundedness_U+}. Hence, the existence of (at least) one $\theta'\in\Theta$ so that $S^{\theta'}$ satisfies $(\mathrm{CPS}^{\lambda'})$ locally  for some $0<\lambda'<\lambda$ is enough to obtain the optimal strategy and to conclude the proof.
\end{proof}

\begin{appendices}
\section{Proofs for Section~\ref{sec:proof_main_result}}\label{app:proof_con_compct}
We begin with the crucial result to prove Lemma \ref{lem:bipolar_relation_model_based}. It was initially proven in \citet{czichowsky2016} (see Lemma~A.1), but under the stronger assumption $(\mathrm{CPS}^{\lambda'})$ locally for all $\lambda'\in(0,\lambda)$. We call a set $\mathcal{G}\subseteq L_{+}^{0}(P)$ \emph{solid}, if $0\leq f\leq g$ and $g\in\mathcal{G}$ imply that $f\in\mathcal{G}$. We further denote $\mathcal{C}^{\theta}\coloneqq\mathcal{C}^{\theta}(1)$ and $\mathcal{D}^{\theta}=\mathcal{D}^{\theta}(1)$, where $\mathcal{C}^{\theta}(x)$ and $\mathcal{D}^{\theta}(y)$ are as defined in \eqref{def:C} and \eqref{def:D}, respectively.
\begin{lemma}\label{lem:C(x)_convex_solid_closed}
    Suppose that $S^{\theta}$ satisfies $(\mathrm{CPS}^{\lambda'})$ locally for some $\lambda'\in(0,\lambda)$. Then, the sets $\mathcal{C}^{\theta}$ and $\mathcal{D}^{\theta}$
    have the following properties:
    \begin{enumerate}
        \item $\mathcal{C}^{\theta}$ and $\mathcal{D}^{\theta}$ are subsets of $L_{+}^{0}(\Omega,\cF,P)$ which are convex, solid and closed in the topology of convergence in measure.\label{itm:conv_solid_clsd}
        \item $\mathcal{C}^{\theta}$ and $\mathcal{D}^{\theta}$ satisfy the bipolarity relation
        \begin{align*}
            g\in\mathcal{C}^{\theta}\iff\E[gh]\leq 1\quad\forall h\in\mathcal{D}^{\theta},\\
            h\in\mathcal{D}^{\theta}\iff\E[gh]\leq 1\quad\forall g\in\mathcal{C}^{\theta}.
        \end{align*}\label{itm:bipolarity}
        \item $\mathcal{C}^{\theta}$ is a bounded subset of $L^{0}(\Omega,\cF,P)$ and contains the constant function $1$.\label{itm:boundedness}
    \end{enumerate}
\end{lemma}
\begin{proof}
    We start with the proof of \eqref{itm:conv_solid_clsd}. The set $\mathcal{C}^{\theta}$ is convex and solid by definition. To prove the closedness of $\mathcal{C}^{\theta}$, let $(H^{n})_{n\in\mathbb{N}}=(H^{0,n},H^{1,n})_{n\in\mathbb{N}}\subseteq\mathcal{H}^{\theta}(1)$ be a sequence of admissible, self-financing trading strategies, such that the sequence $(g_{n})_{n\in\mathbb{N}}\subseteq\mathcal{C}^{\theta}$ with $g_{n}\coloneqq V_{T}^{\mathrm{liq}}(\theta,H^{n})$ converges to some $g\in L_{+}^{0}(P)$ as $n\to\infty$. By passing to a dominating pair $(H^{0,n},H^{1,n})_{n\in\mathbb{N}}$ where equality holds true in equation~\eqref{eq:selffin}, Proposition~\ref{prop:cc_general} together with Lemma~\ref{lem:conint} (assuming $\Theta$ consists only of one model $\theta$) guarantees the existence of a sequence of convex combinations $(\widetilde{H}^{0,n},\widetilde{H}^{1,n})\in\mathrm{conv}((H^{0,n},H^{1,n}),(H^{0,n+1},H^{1,n+1}),\dots)$ and a trading strategy $H=(H^{0},H^{1})\in\mathcal{H}^{\theta}(1)$ satisfying
    \begin{align*}
        P[(\widetilde{H}_{t}^{0,n},\widetilde{H}_{t}^{1,n})\to(H_{t}^{0},H_{t}^{1}),\>\forall t\in[0,T]]=1.
    \end{align*}
    Since $V_{T}^{\mathrm{liq}}(\theta,H^{n})\leq V_{T}^{\mathrm{liq}}(\theta,H^{1,n})$, where $V_{T}^{\mathrm{liq}}(\theta,H^{1,n})$ describes the liquidation value at time $T$ for a dominating pair $(H^{0,n},H^{1,n})$, i.e., where $H^{0,n}$ satisfies equation~\eqref{eq:H0_theta}, we thus obtain
    \begin{align*}
        g&=\lim_{n\to\infty}g_{n}=\lim_{n\to\infty}V_{T}^{\mathrm{liq}}(\theta,H^{n})\leq\lim_{n\to\infty}V_{T}^{\mathrm{liq}}(\theta,H^{1,n})\leq\lim_{n\to\infty}V_{T}^{\mathrm{liq}}(\theta,\widetilde{H}^{1,n})=V_{T}(\theta,H^{1}),
    \end{align*}
    and therefore $g\in\mathcal{C}^{\theta}$. This proves that $\mathcal{C}^{\theta}$ is closed.

    For $\mathcal{D}^{\theta}$ it is enough to note that this set is exactly the polar $(\mathcal{C}^{\theta})^{\circ}$ of $\mathcal{C}^{\theta}$ in $L_{+}^{0}(P)$ as defined in Definition~1.2 of \cite{brannath1999}. That is, for a subset $G\subseteq L_{+}^{0}(P)$, we define its polar by $G^{\circ}=\{h\in L_{+}^{0}(P)\colon\E[gh]\leq 1\text{ for all }g\in G\}$. Hence, it follows from the Bipolar Theorem in $L_{+}^{0}(P)$ by \citet{brannath1999} (see Theorem~1.3) that $\mathcal{D}^{\theta}$ is closed, convex and solid.

    Let us now prove assertion \eqref{itm:bipolarity}. As before, we use that the set $\mathcal{D}^{\theta}$ is the polar  $(\mathcal{C}^{\theta})^{\circ}$ of $\mathcal{C}^{\theta}$ in $L_{+}^{0}(P)$. Hence, by the Bipolar Theorem~1.3 in \cite{brannath1999} the set
    \begin{align*}
        \{g\in L_{+}^{0}(\Omega,\cF,P)\colon\E[gh]\leq 1\>\text{ for each }h\in\mathcal{D}^{\theta}\}
    \end{align*}
    is the smallest closed, convex and solid set in $L_{+}^{0}(\Omega,\cF,P)$ that contains $\mathcal{C}^{\theta}$. Because $\mathcal{C}^{\theta}$ is a convex, closed and solid subset of $L_{+}^{0}(P)$ by assertion~\eqref{itm:conv_solid_clsd}, it follows from the Bipolar Theorem in $L_{+}^{0}(P)$ that $\mathcal{C}^{\theta}$ coincides with its bipolar $(\mathcal{C}^{\theta})^{\circ\circ}$. Since $\mathcal{D}^{\theta}$ is the polar $(\mathcal{C}^{\theta})^{\circ}$, we therefore have that $g\in L_{+}^{0}(P)$ is in $\mathcal{C}^{\theta}$, if and only if $\E[gh]\leq 1$ for all $h\in\mathcal{D}^{\theta}$.
    The second statement in \eqref{itm:bipolarity}, i.e., $h\in\mathcal{D}^{\theta}\iff\E[gh]\leq 1$ for all $g\in\mathcal{C}^{\theta}$, follows by the definition of $\mathcal{D}^{\theta}$.

    It remains to prove assertion~\eqref{itm:boundedness}. The fact that $\mathcal{C}^{\theta}$ contains the constant function $1$ is an immediate consequence of its definition. In particular, we can always take a trading strategy that does not trade in the risky asset, such that $V_{T}^{\mathrm{liq}}(\theta,H)=1$. The $L^{0}(P)$-boundedness of $\mathcal{C}^{\theta}$ is by assertion~\eqref{itm:bipolarity} equivalent to the existence of a strictly positive element $g\in\mathcal{D}^{\theta}$. Hence, it is enough to note that $\{Z^0_T\colon(Z^0,Z^1)\in\mathcal{Z}_{\mathrm{loc}}^{\theta}\}\subseteq\mathcal{D}^{\theta}$, which follows by Assumption~\ref{assump:CPS_all}.
\end{proof}
We continue with a result that was originally proven in \citet{campi2006} in the setting of \citet{kabanov2009}. It provides an a posteriori, quantitative control on the size of the total variation of admissible trading strategies. We use a slightly adjusted version of this result to cover model independent trading strategies. Our proof is mainly the same as the one of Lemma 3.1 in \cite{schachermayer2014} (see also \cite{schachermayer2017}, Lemma 4.10). We also refer to Remark 3.2 in \cite{schachermayer2014} which comes together with the following result.
\begin{prop}\label{prop:bounded_var_Ax}
    Let $x>0$. For some $\theta'\in\Theta$ and for some $0<\lambda'<\lambda$, assume that $S^{\theta'}$ satisfies $(\mathrm{CPS}^{\lambda'})$ locally. Then, for any strategy $H^{1}\in\cA(x)$ with canonical decomposition $H^{1}=H^{1,\uparrow}-H^{1,\downarrow}$, the elements $H_{T}^{1,\uparrow}$ and $H_{T}^{1,\downarrow}$ as well as their convex combinations are bounded in $L^{0}(\Omega,\cF,P)$.
\end{prop}
\begin{proof}
    Fix $0<\lambda'<\lambda$. For some $\theta'\in\Theta$ there is by assumption a probability measure $Q^{\theta'}\approx P$ and a local $Q^{\theta'}$-martingale $\widetilde{S}^{\theta'}=(\widetilde{S}_{t}^{\theta'})_{0\leq t\leq T}$ satisfying \eqref{eq:CPS}. By stopping, we may assume that $\widetilde{S}^{\theta'}$ is a true martingale.

    We consider a strategy $H^{1}\in\cA(x)$, for some $x>0$. Without loss of generality we assume that $H_{T}^{1}=0$, i.e., that the position is liquidated at time $T$. For each $\theta\in\Theta$, using \eqref{eq:H0_theta}, we obtain the holdings of the bond $H_{t}^{0,\theta}=x+H_{t}^{0,\theta,\uparrow}-H_{t}^{0,\theta,\downarrow}$ via $\diff H_{t}^{0,\theta,\uparrow}=(1-\lambda)S_{t}^{\theta}\diff H_{t}^{1,\downarrow}$ and $\diff H_{t}^{0,\theta,\downarrow}=S_{t}^{\theta}\diff H_{t}^{1,\uparrow}$ with $H_{0}^{0,\theta,\uparrow}=H_{0}^{0,\theta,\downarrow}=0$. Now, working with $\theta'$ as introduced before, we first show that
    \begin{equation}\label{eq:L1_bounded_Q}
        \E_{Q^{\theta'}}\big[H_{T}^{0,\theta',\uparrow}\big]\leq\frac{x}{\lambda-\lambda'}.
    \end{equation}
    For this purpose, define the process $H'=((H^{0,\theta'})',(H^{1})')$ by
    \begin{align*}
        H'_{t}=((H^{0,\theta'})'_{t},(H^{1})'_{t})=\Big(H_{t}^{0,\theta'}+\frac{\lambda-\lambda'}{1-\lambda}H_{t}^{0,\theta',\uparrow},H_{t}^{1}\Big),\qquad 0\leq t\leq T.
    \end{align*}
    This process is a self-financing trading strategy under transaction costs $\lambda'$: indeed, whenever we have $\diff H_{t}^{0,\theta'}>0$ such that $\diff H_{t}^{0,\theta'}=\diff H_{t}^{0,\theta',\uparrow}$, the agent sells some units of stock and receives $\diff H_{t}^{0,\theta',\uparrow}=(1-\lambda)S_{t}^{\theta'}\diff H_{t}^{1,\downarrow}$ (resp., $(1-\lambda')S_{t}^{\theta'}\diff H_{t}^{1,\downarrow}=\frac{1-\lambda'}{1-\lambda}\diff H_{t}^{0,\theta',\uparrow}$) many bonds under transaction costs $\lambda$ (resp., $\lambda'$). The difference between these two terms is $\frac{\lambda-\lambda'}{1-\lambda}\diff H_{t}^{0,\theta',\uparrow}$; this difference is the amount by which the $\lambda'$-agent does better than the $\lambda$-agent. It is also clear that $((H^{0,\theta'})',(H^{1})')$ under transaction costs $\lambda'$ still is an $x$-admissible strategy for the model $\theta'$, i.e., $H'\in\mathcal{H}^{\theta'}(x)$.

    By Proposition \ref{prop:strong_supermartingale}, the process $\widetilde{V}(\theta',H')=(\widetilde{V}_{t}(\theta',H'))_{0\leq t\leq T}$ defined by
    \begin{align*}
        \widetilde{V}_{t}(\theta',H')=(H^{0,\theta'})'_{t}+(H^{1})'_{t}\widetilde{S}_{t}^{\theta'}=H_{t}^{0,\theta'}+\frac{\lambda-\lambda'}{1-\lambda}H_{t}^{0,\theta',\uparrow}+H_{t}^{1}\widetilde{S}_{t}^{\theta'}
    \end{align*}
    is an optional strong $Q^{\theta'}$-supermartingale. We thus get
    \begin{align*}
        \E_{Q^{\theta'}}\big[\widetilde{V}_{T}\big(\theta',H'\big)\big]=\E_{Q^{\theta'}}\big[H_{T}^{0,\theta'}+H_{T}^{1}\widetilde{S}_{T}^{\theta'}\big]+\frac{\lambda-\lambda'}{1-\lambda}\E_{Q^{\theta'}}\big[H_{T}^{0,\theta',\uparrow}\big]\leq x.
    \end{align*}
    By admissibility of $H^{1}$, we have $H_{T}^{0,\theta'}+H_{T}^{1}\widetilde{S}_{T}^{\theta'}\geq 0$, and thus we have shown \eqref{eq:L1_bounded_Q}.

    To obtain control on $H_{T}^{0,\theta',\downarrow}$ too, we note that $H_{T}^{0,\theta'}\geq -x$, since $H_{T}^{1}=0$. So we have $H_{T}^{0,\theta',\downarrow}\leq x+H_{T}^{0,\theta',\uparrow}$. Therefore, we obtain the following estimate for the total variation $H_{T}^{0,\theta',\uparrow}+H_{T}^{0,\theta',\downarrow}$ of $H^{0,\theta'}$:
    \begin{align}\label{eq:L1Q_estimate}
        \E_{Q^{\theta'}}\big[H_{T}^{0,\theta',\uparrow}+H_{T}^{0,\theta',\downarrow}\big]\leq x\Big(1+\frac{2}{\lambda-\lambda'}\Big).
    \end{align}
    Now we transfer the $L^{1}(Q^{\theta'})$-estimate in \eqref{eq:L1Q_estimate} to an $L^{0}(P)$-estimate. For $\varepsilon>0$, there exists $\delta_{\theta'}>0$, so that for $A\in\mathcal{F}$ with $Q^{\theta'}[A]<\delta_{\theta'}$, we have $P[A]<\frac{\varepsilon}{2}$. Letting $C^{0,\theta'}\coloneqq\frac{x}{\delta_{\theta'}}\big(1+\frac{2}{\lambda-\lambda'}\big)$ and applying Markov's inequality to \eqref{eq:L1Q_estimate}, we get
    \begin{align}\label{eq:L0P_estimate}
        P\big[H_{T}^{0,\theta',\uparrow}+H_{T}^{0,\theta',\downarrow}\geq C^{0,\theta'}\big]<\frac{\varepsilon}{2},
    \end{align}
    which is the desired $L^{0}(P)$-estimate. At this point we remark that \eqref{eq:L1_bounded_Q} implies that the convex hull of the functions $H_{T}^{0,\theta',\uparrow}$ is bounded in $L^{1}(Q^{\theta'})$ and \eqref{eq:L1Q_estimate} yields the same for $H_{T}^{0,\theta',\downarrow}$. So by the above reasoning we obtain that also the convex combinations of $H_{T}^{0,\theta',\uparrow}$ and $H_{T}^{0,\theta',\downarrow}$ remain bounded in $L^{0}(P)$.\\

    As before, it follows from \eqref{eq:H0_theta} that $\diff H_{t}^{0,\theta',\downarrow}=S_{t}^{\theta'}\diff H_{t}^{1,\uparrow}$, which can be rewritten as
    \begin{equation}\label{eq:control_H1}
        \diff H_{t}^{1,\uparrow}=\frac{\diff H_{t}^{0,\theta',\downarrow}}{S_{t}^{\theta'}}.
    \end{equation}
    By assumption, the trajectories of $S^{\theta'}$ are strictly positive. In fact, we even have for almost all trajectories $(S_{t}^{\theta'}(\omega))_{0\leq t\leq T}$, that $\inf_{0\leq t\leq T} S_{t}^{\theta'}(\omega)$ is strictly positive. Indeed, $\widetilde{S}^{\theta'}$ being a $Q^{\theta'}$-martingale with $\widetilde{S}_{T}^{\theta'}>0$ almost surely satisfies that $\inf_{0\leq t\leq T}\widetilde{S}_{t}^{\theta'}$ is $Q^{\theta'}$-a.s. and therefore $P$-a.s. strictly positive. In particular, for $\varepsilon>0$ we may find $\delta'_{\theta'}>0$ such that
    \begin{equation}\label{eq:S_nonnegative}
        P\Big[\inf_{0\leq t\leq T}S_{t}^{\theta'}<\delta'_{\theta'}\Big]<\frac{\varepsilon}{2}.
    \end{equation}
    Taking $\eta_{\theta'}\coloneqq\min(\delta_{\theta'},\delta'_{\theta'})$ and letting $C^{1,\theta'}\coloneqq\frac{x}{\eta_{\theta'}^{2}}\big(1+\frac{2}{\lambda-\lambda'}\big)$, we obtain from \eqref{eq:L0P_estimate}, \eqref{eq:control_H1} and \eqref{eq:S_nonnegative} that
    \begin{equation}\label{eq:L0P_H1_theta_estimate}
        P\Big[H_{T}^{1,\uparrow}\geq C^{1,\theta'}\Big]\leq P\Big[\inf_{0\leq t\leq T}S_{t}^{\theta'}<\delta'_{\theta'}\Big]+P\big[H_{T}^{0,\theta',\uparrow}+H_{T}^{0,\theta',\downarrow}\geq C^{0,\theta'}\big]<\varepsilon.
    \end{equation}
    The first inequality in \eqref{eq:L0P_H1_theta_estimate} follows from
    \begin{align*}
        &P\Big[H_{T}^{1,\uparrow}\geq C^{1,\theta'}\Big]\\
        &\qquad\leq P\Big[H_{T}^{0,\theta',\downarrow}\geq \Big(\inf_{0\leq t\leq T}S_{t}^{\theta'}\Big)C^{1,\theta'}\Big]\\
        &\qquad\leq P\Big[\inf_{0\leq t\leq T}S_{t}^{\theta'}<\delta'_{\theta'}\Big]+P\Big[\Big\{H_{T}^{0,\theta',\downarrow}\geq \Big(\inf_{0\leq t\leq T}S_{t}^{\theta'}\Big)C^{1,\theta'}\Big\}\cap\Big\{\inf_{0\leq t\leq T}S_{t}^{\theta'}\geq\delta'_{\theta'}\Big\}\Big],
    \end{align*}
    which uses \eqref{eq:control_H1} (note that $H_{0}^{0,\theta',\uparrow}=H_{0}^{0,\theta',\downarrow}=0$ by definition) together with
    \begin{align*}
        &P\Big[\Big\{H_{T}^{0,\theta',\downarrow}\geq \Big(\inf_{0\leq t\leq T}S_{t}^{\theta'}\Big)C^{1,\theta'}\Big\}\cap\Big\{\inf_{0\leq t\leq T}S_{t}^{\theta'}\geq\delta'_{\theta'}\Big\}\Big]\\
        &\qquad\leq P\Big[H_{T}^{0,\theta',\downarrow}\geq\delta'_{\theta'}C^{1,\theta'}\Big]\leq P\Big[H_{T}^{0,\theta',\downarrow}\geq C^{0,\theta'}\Big]\leq P\Big[H_{T}^{0,\theta',\uparrow}+H_{T}^{0,\theta',\downarrow}\geq C^{0,\theta'}\Big],
    \end{align*}
    where we use that $\eta_{\theta'}^{2}\leq\delta_{\theta'}\delta'_{\theta'}$, which implies $\delta'_{\theta'}C^{1,\theta'}\geq C^{0,\theta'}$. The second inequality in \eqref{eq:L0P_H1_theta_estimate} then follows from \eqref{eq:L0P_estimate} and \eqref{eq:S_nonnegative}.

    To control the term $H_{T}^{1,\downarrow}$, we observe that $H_{T}^{1,\uparrow}-H_{T}^{1,\downarrow}=H_{T}^{1}=0$. Therefore, we may use the estimate \eqref{eq:L0P_H1_theta_estimate} of $H_{T}^{1,\uparrow}$ to also control $H_{T}^{1,\downarrow}$. Moreover, we note that \eqref{eq:control_H1} also holds for convex combinations of $H^{1,\uparrow}$. Indeed, for another strategy $\widehat{H}^{1}\in\cA(x)$ and $\alpha\in[0,1]$ we have
    \begin{align*}
        S_{t}^{\theta'}\diff((1-\alpha)H_{t}^{1,\uparrow}+\alpha\widehat{H}_{t}^{1,\uparrow})&=(1-\alpha)S_{t}^{\theta'}\diff H_{t}^{1,\uparrow}+\alpha S_{t}^{\theta'}\diff\widehat{H}_{t}^{1,\uparrow}=(1-\alpha)\diff H_{t}^{0,\theta',\downarrow}+\alpha\diff\widehat{H}_{t}^{0,\theta',\downarrow},
    \end{align*}
    so that dividing by $S_{t}^{\theta'}$ yields
    \begin{align*}
        \diff((1-\alpha)H_{t}^{1,\uparrow}+\alpha\widehat{H}_{t}^{1,\uparrow})&=\frac{\diff((1-\alpha)H_{t}^{0,\theta',\downarrow}+\alpha\widehat{H}_{t}^{0,\theta',\downarrow})}{S_{t}^{\theta'}}.
    \end{align*}
    Since \eqref{eq:L0P_estimate} also holds for convex combinations of $H_{T}^{0,\theta',\uparrow}$ and $H_{T}^{0,\theta',\downarrow}$, we obtain that also the convex combinations of $H_{T}^{1,\uparrow}$ and $H_{T}^{1,\downarrow}$ remain bounded in $L^{0}(P)$.
\end{proof}
Next, we establish the proof of Lemma \ref{lem:conint}.
\begin{proof}[Proof of Lemma \ref{lem:conint}]
    Let $\theta\in\Theta$ and fix an arbitrary $\varepsilon>0$. Also fix $\omega\in\Omega$ so that $S^{\theta}(\omega)$ has c{\`a}dl{\`a}g trajectories and such that $H_{t}^{1,n,\uparrow}(\omega)\to H_{t}^{1,\uparrow}(\omega)$ for all $t\in[0,T]$. Since the function $S^{\theta}_{\cdot}(\omega):[0,T]\to\R$ is c{\`a}dl{\`a}g, there can only be finitely many times $0\leq \tau_1<\ldots<\tau_k\leq T$ such that $\big|\Delta S^{\theta}_{\tau_i}(\omega)\big|\geq \varepsilon$. Indeed, if there were infinitely many time points with jumps of size larger than $\varepsilon$, the jump times would have a cluster point in the compact set $[0,T]$, leading to a contradiction to the existence of right or left limits of $S^{\theta}_t(\omega)$ at every $t\in[0,T]$.
    
    Therefore, setting
    \begin{align*}
        S_{t}^{\theta,\varepsilon}(\omega)\coloneqq S_{t}^{\theta}(\omega)-\sum_{i=1}^{k}\Delta S_{\tau_i}^{\theta}(\omega)\mathbbm{1}_{[\tau_i,T]}(t),\quad t\in[0,T],
    \end{align*}
    gives a c{\`a}dl{\`a}g function with $\big|\Delta S^{\theta,\varepsilon}_t(\omega)\big|\leq\varepsilon$ for all $t\in[0,T]$. Then, there are finitely many times $\sigma_0=0<\sigma_1<\ldots<\sigma_{m-1}<T$ such that $\big|S_{\sigma_{i+1}}^{\theta,\varepsilon}(\omega)-S_{\sigma_{i}}^{\theta,\varepsilon}(\omega)\big|>\varepsilon$, because $S^{\theta,\varepsilon}_{\cdot}(\omega):[0,T]\to\R$ is c{\`a}dl{\`a}g. Indeed, if there were infinitely many such time points $(\sigma_i)_{i\in I}$ for some infinite index set $I$, the times $(\sigma_i)_{i\in I}$ would have a cluster point in the compact set $[0,T]$ leading to a contradiction to the existence of right or left limits of $S^{\theta,\varepsilon}_t(\omega)$ at every $t\in[0,T]$. Because the jumps of $S^{\theta,\varepsilon}(\omega)$ are bounded by $\varepsilon$, we obtain $\big|S_{t}^{\theta,\varepsilon}(\omega)-S_{\sigma_{i}}^{\theta,\varepsilon}(\omega)\big|\leq 2\varepsilon$ for all $t\in[\sigma_{i},\sigma_{i+1}]$ for $i=0,\ldots,m-1$ with $\sigma_m:=T$. Therefore, the step functions $S^{\theta,\varepsilon,m}(\omega)$ given by 
    \begin{align*}
        S_{t}^{\theta,\varepsilon,m}(\omega)&\coloneqq S_{0}^{\theta,\varepsilon}(\omega)\mathbbm{1}_{\{0\}}(t)+\sum_{i=1}^{m}S_{\sigma_{i-1}}^{\theta,\varepsilon}(\omega)\mathbbm{1}_{(\sigma_{i-1},\sigma_{i}]}(t),\quad t\in[0,T],
    \end{align*}
satisfy $\big|S_{t}^{\theta,\varepsilon}(\omega)-S_{t}^{\theta,\varepsilon,m}(\omega)\big|\leq 2\varepsilon$ for all $t\in[0,T]$, which implies
    \begin{align*}
        \bigg\vert\int_{0}^{t}\big(S_{u}^{\theta,\varepsilon}(\omega)-S_{u}^{\theta,\varepsilon,m}(\omega)\big)\diff H_{u}^{1,\uparrow,n}(\omega)\bigg\vert&\leq 2\varepsilon H_{T}^{1,\uparrow,n}(\omega),\\
        \bigg\vert\int_{0}^{t}\big(S_{u}^{\theta,\varepsilon}(\omega)-S_{u}^{\theta,\varepsilon,m}(\omega)\big)\diff H_{u}^{1,\uparrow}(\omega)\bigg\vert&\leq 2\varepsilon H_{T}^{1,\uparrow}(\omega).
    \end{align*}
    Moreover, because of our definition of the stochastic integral at left jumps of the integrator (see \eqref{eq:stoch_int_up} and \eqref{eq:stoch_int_down}), we have that
    \begin{align*}
        &\int_{0}^{t}S_{u}^{\theta,\varepsilon,m}(\omega)\diff H_{u}^{1,n,\uparrow}(\omega)\\
        &\qquad=\sum_{i=1}^{m}\int_{(\sigma_{i-1}\wedge t)+}^{\sigma_{i}\wedge t}S_{u}^{\theta,\varepsilon,m}(\omega)\diff H_{u}^{1,n,\uparrow}(\omega)+\sum_{i=1}^{m}S_{\sigma_{i-1}}^{\theta,\varepsilon,m}(\omega)\Delta_{+}H_{\sigma_{i-1}}^{1,n,\uparrow}(\omega)\mathbbm{1}_{\{\sigma_{i-1}<t\}}\\
        &\qquad=\sum_{i=1}^{m}S_{\sigma_{i-1}}^{\theta,\varepsilon}(\omega)(H_{\sigma_{i}\wedge t}^{1,n,\uparrow}(\omega)-H_{(\sigma_{i-1}\wedge t)+}^{1,n,\uparrow}(\omega))+S_{0}^{\theta,\varepsilon}(\omega)\Delta_{+}H_{0}^{1,n,\uparrow}(\omega)\mathbbm{1}_{\{t>0\}}\\
        &\qquad+\sum_{i=2}^{m}S_{\sigma_{i-2}}^{\theta,\varepsilon}(\omega)\Delta_{+}H_{\sigma_{i-1}}^{1,n,\uparrow}(\omega)\mathbbm{1}_{\{\sigma_{i-1}<t\}}\\
        &\qquad=\sum_{i=1}^{m}S_{\sigma_{i-1}}^{\theta,\varepsilon}(\omega)(H_{\sigma_{i}\wedge t}^{1,n,\uparrow}(\omega)-H_{\sigma_{i-1}\wedge t}^{1,n,\uparrow}(\omega))-\sum_{i=1}^{m-1}(S_{\sigma_{i}}^{\theta,\varepsilon}(\omega)-S_{\sigma_{i-1}}^{\theta,\varepsilon}(\omega))\Delta_{+}H_{\sigma_{i}}^{1,n,\uparrow}(\omega)\mathbbm{1}_{\{\sigma_{i}<t\}},
    \end{align*}
    where we use that $\int_{(\sigma_{i-1}\wedge t)+}^{\sigma_{i}\wedge t}S_{u}^{\theta,\varepsilon,m}(\omega)\diff H_{u}^{1,n,\uparrow}(\omega)=S_{\sigma_{i-1}}^{\theta,\varepsilon}(\omega)(H_{\sigma_{i}\wedge t}^{1,n,\uparrow}(\omega)-H_{(\sigma_{i-1}\wedge t)+}^{1,n,\uparrow}(\omega))$, as well as $\Delta_{+}H_{\sigma_{i-1}\wedge t}^{1,n,\uparrow}(\omega)=H_{(\sigma_{i-1}\wedge t)+}^{1,n,\uparrow}(\omega)-H_{\sigma_{i-1}\wedge t}^{1,n,\uparrow}(\omega)$, by definition. Similarly, we obtain
    \begin{align*}
        &\int_{0}^{t}S_{u}^{\theta,\varepsilon,m}(\omega)\diff H_{u}^{1,\uparrow}(\omega)\\
        &\qquad=\sum_{i=1}^{m}S_{\sigma_{i-1}}^{\theta,\varepsilon}(\omega)\big(H_{\sigma_{i}\wedge t}^{1,\uparrow}(\omega)-H_{\sigma_{i-1}\wedge t}^{1,\uparrow}(\omega)\big)-\sum_{i=1}^{m-1}(S_{\sigma_{i}}^{\theta,\varepsilon}(\omega)-S_{\sigma_{i-1}}^{\theta,\varepsilon}(\omega))\Delta_{+}H_{\sigma_{i}}^{1,\uparrow}(\omega)\mathbbm{1}_{\{\sigma_{i}<t\}},
    \end{align*}
    so that
    \begin{align*}
        \int_{0}^{t}S_{u}^{\theta,\varepsilon,m}(\omega)\diff H_{u}^{1,n,\uparrow}(\omega)\to\int_{0}^{t}S_{u}^{\theta,\varepsilon,m}(\omega)\diff H_{u}^{1,\uparrow}(\omega),\quad n\to\infty,
    \end{align*}
     as $H_{t}^{1,n,\uparrow}(\omega)\to H_{t}^{1,\uparrow}(\omega)$ for all $t\in[0,T]$.
    In particular, for $n\in\mathbb{N}$ large enough, we have
    \begin{align*}
        &\bigg\vert\int_{0}^{t}S_{u}^{\theta,\varepsilon}(\omega)\diff H_{u}^{1,n,\uparrow}(\omega)-\int_{0}^{t}S_{u}^{\theta,\varepsilon}(\omega)\diff H_{u}^{1,\uparrow}(\omega)\bigg\vert\\
        &\qquad\leq\bigg\vert\int_{0}^{t}\big(S_{u}^{\theta,\varepsilon}(\omega)-S_{u}^{\theta,\varepsilon,m}(\omega)\big)\diff H_{u}^{1,n,\uparrow}(\omega)\bigg\vert+\bigg\vert\int_{0}^{t}\big(S_{u}^{\theta,\varepsilon,m}(\omega)-S_{u}^{\theta,\varepsilon}(\omega)\big)\diff H_{u}^{1,\uparrow}(\omega)\bigg\vert\\
        &\qquad\quad+\bigg\vert\int_{0}^{t}S_{u}^{\theta,\varepsilon,m}(\omega)\diff H_{u}^{1,n,\uparrow}(\omega)-\int_{0}^{t}S_{u}^{\theta,\varepsilon,m}(\omega)\diff H_{u}^{1,\uparrow}(\omega)\bigg\vert\\
        &\qquad\leq 4\varepsilon\big(H_{T}^{1,\uparrow}(\omega)+\varepsilon\big)+\varepsilon.
    \end{align*}
   
   On the other hand, the finite sum $\sum_{i=1}^k\Delta S_{\tau_i}^{\theta}(\omega)\mathbbm{1}_{[\tau_i,T]}(t)$, where the $\tau_{i}$'s are the times where $\big|\Delta S^{\theta}_{\tau_i}(\omega)\big|\geq\varepsilon$ for $i=1,\ldots,k$, satisfies
    \begin{align}
        \int_{0}^{t}\bigg(\sum_{i=1}^k\Delta S_{\tau_i}^{\theta}(\omega)\mathbbm{1}_{[\tau_i,T]}(u)\bigg)\diff H_{u}^{1,n,\uparrow}(\omega)&=\sum_{i=1}^{k}\Delta S_{\tau_{i}}^{\theta}(\omega)\big(H_{t}^{1,n,\uparrow}(\omega)-H_{\tau_{i}\wedge t}^{1,n,\uparrow}(\omega)\big),\label{eq:intdeltaS:1}
    \end{align}
    and
        \begin{align}
        \int_{0}^{t}\bigg(\sum_{i=1}^k\Delta S_{\tau_i}^{\theta}(\omega)\mathbbm{1}_{[\tau_i,T]}(u)\bigg)\diff H_{u}^{1,\uparrow}(\omega)&=\sum_{i=1}^{k}\Delta S_{\tau_{i}}^{\theta}(\omega)\big(H_{t}^{1,\uparrow}(\omega)-H_{\tau_{i}\wedge t}^{1,\uparrow}(\omega)\big).\label{eq:intdeltaS:2}
    \end{align}
Here, we again exploit our definition of the stochastic integral at left jumps of the integrator (see \eqref{eq:stoch_int_up} and \eqref{eq:stoch_int_down}). Since $H_{t}^{1,n,\uparrow}(\omega)\to H_{t}^{1,\uparrow}(\omega)$ for all $t\in[0,T]$, we thus have that
    \begin{align*}
        \int_{0}^{t}\bigg(\sum_{i=1}^k\Delta S_{\tau_i}^{\theta}(\omega)\mathbbm{1}_{[\tau_i,T]}(u)\bigg)\diff H_{u}^{1,n,\uparrow}(\omega)\to\int_{0}^{t}\bigg(\sum_{i=1}^k\Delta S_{\tau_i}^{\theta}(\omega)\mathbbm{1}_{[\tau_i,T]}(u)\bigg)\diff H_{u}^{1,\uparrow}(\omega),
    \end{align*}
    as $n\to\infty$, by \eqref{eq:intdeltaS:1} and \eqref{eq:intdeltaS:2}. Together with the above, it thus follows that
    \begin{align*}
       \int_{0}^{t}S_{u}^{\theta}(\omega)\diff H_{u}^{1,n,\uparrow}(\omega)\to\int_{0}^{t}S_{u}^{\theta}(\omega)\diff H_{u}^{1,\uparrow}(\omega),\quad n\to\infty,
    \end{align*}
    for every $t\in[0,T]$, since $\varepsilon>0$ was arbitrary.
    
    Repeating the same argument as above for $H^{1,n,\downarrow}$, $H^{1,\downarrow}$ and $(1-\lambda)S^{\theta}$, we also obtain
    \begin{align*}
       \int_{0}^{t}S_{u}^{\theta}(\omega)\diff H_{u}^{1,n,\downarrow}(\omega)\to\int_{0}^{t}S_{u}^{\theta}(\omega)\diff H_{u}^{1,\downarrow}(\omega),\quad n\to\infty.
    \end{align*}
    Since $\omega$ can be chosen arbitrarily from a set with probability $1$, this proves \eqref{eq:conv_H0_decomp}.
    
    In order to prove that $H^{1}=H^{1,\uparrow}\mathbbm{1}_{\llbracket 0,\sigma_{1}\wedge\sigma_{2}\llbracket}\mathbbm{1}_{\llbracket0,\sigma_{3}\rrbracket}-H^{1,\downarrow}\mathbbm{1}_{\llbracket 0,\sigma_{1}\wedge\sigma_{2}\llbracket}\mathbbm{1}_{\llbracket0,\sigma_{3}\rrbracket}\in\mathcal{A}(x)$, we first note that all the $H^{0,\theta,n,\uparrow}$ and $H^{0,\theta,n,\downarrow}$ as defined in \eqref{eq:H_0_n_down_and_H_0_n_up} are $[0,\infty]$-valued, increasing and $\mathbb{F}$-predictable processes.
    This consequence follows from our definition of the stochastic integral and the fact that $H^{1,n,\uparrow}$ as well as $H^{1,n,\downarrow}$ are $[0,\infty]$-valued, increasing and $\mathbb{F}$-predictable processes, together with $S^{\theta}$ being strictly positive and $\mathbb{F}$-adapted. Since $H^{1,n}\in\cA(x)$ for each $n\in\mathbb{N}$, it follows from the proof of Proposition~\ref{prop:bounded_var_Ax} that $P[H_{T}^{0,\theta,n,\uparrow}<\infty]=1$ and $P[H_{T}^{0,\theta,n,\downarrow}<\infty]=1$. Moreover, for arbitrary $M\geq 0$ and $n\in\mathbb{N}$ we may write
    \begin{align*}
        P[H_{T}^{0,\theta,\uparrow}\geq M]&=P[\{H_{T}^{0,\theta,\uparrow}\geq M\}\cap\{H_{T}^{0,\theta,n,\uparrow}\geq M\}]+P[\{H_{T}^{0,\theta,\uparrow}\geq M\}\cap\{H_{T}^{0,\theta,n,\uparrow}< M\}]\\
        &\leq P[H_{T}^{0,\theta,n,\uparrow}\geq M]+P[H_{T}^{0,\theta,\uparrow}>H_{T}^{0,\theta,n,\uparrow}],
    \end{align*}
    so that letting $M\to\infty$ yields $P[H_{T}^{0,\theta,\uparrow}<\infty]\geq 1-P[H_{T}^{0,\theta,\uparrow}>H_{T}^{0,\theta,n,\uparrow}]$. However, by taking the limit as $n\to\infty$ on both sides, we obtain from equation~\eqref{eq:conv_H0_decomp} that $P[H_{T}^{0,\theta,\uparrow}<\infty]=1$. The same argumentation also yields $P[H_{T}^{0,\theta,\downarrow}<\infty]=1$. Then, from
    \begin{align*}
        &P[\{\omega\in\Omega\colon H_{t}^{0,\theta,\uparrow}(\omega)\mathbbm{1}_{\llbracket 0,\sigma_{1}\wedge\sigma_{2}\llbracket}(\omega,t)\mathbbm{1}_{\llbracket0,\sigma_{3}\rrbracket}(\omega,t)=\infty\}]\\
        &\quad=P[\{H_{t}^{0,\theta,\uparrow}\mathbbm{1}_{0\leq t<\sigma_{1}\wedge\sigma_{2}}\mathbbm{1}_{0\leq t\leq\sigma_{3}}=\infty\}\cap\{\{\sigma_{1}\wedge\sigma_{2}<\infty\}\cup\{\sigma_{3}<\infty\}\}]\\
        &\quad+P[\{H_{t}^{0,\theta,\uparrow}\mathbbm{1}_{0\leq t<\sigma_{1}\wedge\sigma_{2}}\mathbbm{1}_{0\leq t\leq\sigma_{3}}=\infty\}\cap\{\{\sigma_{1}\wedge\sigma_{2}=\infty\}\cap\{\sigma_{3}=\infty\}\}]\\
        &\quad\leq P[\{\sigma_{1}\wedge\sigma_{2}<\infty\}\cup\{\sigma_{3}<\infty\}]+P[H_{t}^{0,\theta,\uparrow}=\infty]=0,
    \end{align*}
    and similarly $P[\{\omega\in\Omega\colon H_{t}^{0,\theta,\downarrow}(\omega)\mathbbm{1}_{\llbracket 0,\sigma_{1}\wedge\sigma_{2}\llbracket}(\omega,t)\mathbbm{1}_{\llbracket0,\sigma_{3}\rrbracket}(\omega,t)=\infty\}]=0$, the limit process $H^{0,\theta}=(H_{t}^{0,\theta})_{0\leq t\leq T}$, defined via
    \begin{align*}
        &H_{t}^{0,\theta}(\omega)\\
        &\quad=x\mathbbm{1}_{\llbracket 0,\sigma_{1}\wedge\sigma_{2}\llbracket\cap\llbracket0,\sigma_{3}\rrbracket}(\omega,t)+H_{t}^{0,\theta,\uparrow}(\omega)\mathbbm{1}_{\llbracket 0,\sigma_{1}\wedge\sigma_{2}\llbracket\cap\llbracket0,\sigma_{3}\rrbracket}(\omega,t)-H_{t}^{0,\theta,\downarrow}(\omega)\mathbbm{1}_{\llbracket 0,\sigma_{1}\wedge\sigma_{2}\llbracket\cap\llbracket0,\sigma_{3}\rrbracket}(\omega,t),
    \end{align*}
    is real-valued and $\mathbb{F}$-predictable. This is a direct consequence of $\sigma_{1}$ and $\sigma_{2}$ being predictable stopping times (see Remark~IV.87(d) in \cite{dellacherie1978} together with Remark~E of the preliminary section ``Complements to Chapter~IV'' in \cite{dellacherie1982}). Its total variation $\vert H^{0,\theta}\vert=(\vert H^{0,\theta}\vert_{t})_{0\leq t\leq T}$ is $[0,\infty]$-valued, $\mathbb{F}$-predictable and $P$-a.s.~finite.
    It thus remains to check that $V^{\mathrm{liq}}(\theta,H^{1})$ satisfies the admissibility condition \eqref{eq:admissible} for all $\theta\in\Theta$. By assumption, the processes $(H^{1,n})_{n\in\mathbb{N}}$ are $x$-admissible for all $\theta\in\Theta$. Hence, by identities \eqref{eq:conv_H1_decomp} and \eqref{eq:conv_H0_decomp}, we get for every $t\in[0,T]$ and for each $\theta\in\Theta$ that $V_{t}^{\mathrm{liq}}(\theta,H^{1,n})\to V_{t}^{\mathrm{liq}}(\theta,H^{1})$ almost surely, by the continuity of the liquidation function \eqref{eq:admissible} with respect to $(H^{0,\theta}_t,H^1_t)$, so that the admissibility condition \eqref{eq:admissible} passes to the limit $H^{1}$. We thus have $H^{1}\in\mathcal{A}(x)$ and this concludes the proof.
\end{proof}

To prove Proposition \ref{prop:cc_general}, we also use the following well-known variant of Koml{\'o}s' theorem (see \cite{delbaen1994}, Lemma A1.1).
\begin{lemma}\label{lem:komlos}
    Let $(f_{n})_{n\in\mathbb{N}}$ be a sequence of $\mathbb{R}_{+}$-valued, random variables on a probability space $(\Omega,\mathcal{F},P)$. There is a sequence $g_{n}\in\mathrm{conv}(f_{n},f_{n+1},\dots)$ of convex combinations that converges almost surely to a $[0,\infty]$-valued function $g$. If $\mathrm{conv}(f_{n},f_{n+1},\dots)$ is bounded in $L^{0}(\Omega,\mathcal{F},P)$, then $g$ is finite almost surely.
\end{lemma}
We further need the next simple fact about the measurability of limits of measurable functions (see, e.g., Lemma 3.5 in \cite{williams1991}).
\begin{lemma}\label{lem:meas}
    Let $(f_{n})_{n\in\mathbb{N}}$ be a sequence of measurable functions on a measure space $(\Omega,\mathcal{F})$. Then, $\liminf_{n\to\infty}f_n$ and $\limsup_{n\to\infty}f_n$ are $[-\infty,\infty]$-valued measurable functions and the set $F:=\left\{\omega\in\Omega: \text{$f_n(\omega)$ converges to a limit in $\mathbb{R}$}\right\}$ is $\mathcal{F}$-measurable and given by
    \begin{align*}
        F={}&\left\{\omega\in\Omega: \liminf_{n\to\infty}f_n(\omega)=\limsup_{n\to\infty}f_n(\omega)\in\mathbb{R}\right\}.
    \end{align*}
    Moreover, if $\liminf_{n\to\infty}f_n$ and $\limsup_{n\to\infty}f_n$ are measurable with respect to a sub-$\sigma$-algebra $\mathcal{G}\subseteq\mathcal{F}$, then $F\in\mathcal{G}$.
\end{lemma}
We are now ready to prove Proposition \ref{prop:cc_general}. Here, the main difference to the initial paper by \citet{chau2019} (see Lemma~A.1) is the treatment of jump times. For this reason, we revisit the proof of \cite{campi2006}, Proposition 3.4\footnote{Note that, since we are already considering the liquidation value in \eqref{eq:primalproblem}, we do not need to assume that $H^1_T=0$. Therefore, we also do not need to assume that $\mathcal{F}_{T}=\mathcal{F}_{T-}$ and $S_{T}^{\theta}=S_{T-}^{\theta}$ as in Remark 4.2 in \cite{campi2006}.}. Our key insight here is that we can show that the set where the convergence can fail, is the same for all models. It can again be exhausted by countably many stopping times. However, since the filtration $\mathbb{F}$ is not assumed to be complete, the treatment of jump times needs special care.
\begin{proof}[Proof of Proposition \ref{prop:cc_general}]
    Fix $x>0$ and let $(H^{1,n})_{n\in\mathbb{N}}\subseteq\mathcal{A}(x)$ be a sequence of admissible, self-financing strategies. In particular, $H^{1,n}=(H_{t}^{1,n})_{0\leq t\leq T}$ is a real-valued, $\mathbb{F}$-predictable stochastic process, whose total variation $\vert H^{1,n}\vert_{T}$ is $[0,\infty]$-valued, $\mathbb{F}$-predictable and satisfies $P[\vert H^{1,n}\vert_{T}<\infty]=1$,
    and $V^{\mathrm{liq}}(\theta,H^{1,n})$ satisfies \eqref{eq:admissible} for all $\theta\in\Theta$. As above, we decompose these processes canonically as $H_{t}^{1,n}(\omega)=H_{t}^{1,n,\uparrow}(\omega)-H_{t}^{1,n,\downarrow}(\omega)$ for almost every $\omega$ for every $t\in[0,T]$, with $H^{1,n,\uparrow}$ and $H^{1,n,\downarrow}$ both being $[0,\infty]$-valued, $\mathbb{F}$-predictable and satisfying $P[H^{1,n,\uparrow}_{T}<\infty]=P[H^{1,n,\downarrow}_{T}<\infty]=1$.
    By Proposition~\ref{prop:bounded_var_Ax} we know that $(H_{T}^{1,n,\uparrow})_{n\in\mathbb{N}}$ and $(H_{T}^{1,n,\downarrow})_{n\in\mathbb{N}}$ as well as their convex combinations are bounded in $L^{0}(\Omega,\mathcal{F},P)$. Hence, let $D\coloneqq\big([0,T]\cap\mathbb{Q}\big)\cup\{T\}$ and use Lemma~\ref{lem:komlos} together with a diagonalisation procedure to obtain sequences of convex weights $\alpha_{n}^{j}$ such that for
    \begin{align*}
        \widetilde{H}_{t}^{1,n,\uparrow}=\sum_{j\geq 1}\alpha_{n}^{j}H_{t}^{1,n+j-1,\uparrow},\quad\widetilde{H}_{t}^{1,n,\downarrow}=\sum_{j\geq 1}\alpha_{n}^{j}H_{t}^{1,n+j-1,\downarrow},\quad t\in D,
    \end{align*}
    there exist $\mathcal{F}_{t}$-measurable random variables $\widetilde{H}_{t}^{1,\uparrow}$ and $\widetilde{H}_{t}^{1,\downarrow}$, such that
    \begin{align}\label{eq:komlos_conv_down}
        \widetilde{H}_{t}^{1,n,\uparrow}\to\widetilde{H}_{t}^{1,\uparrow},\quad\widetilde{H}_{t}^{1,n,\downarrow}\to\widetilde{H}_{t}^{1,\downarrow},\quad \forall t\in D,
    \end{align}
    almost surely. We denote by $\widetilde{\Omega}_{0}$ the event where \eqref{eq:komlos_conv_down} holds true so that $P[\widetilde{\Omega}_{0}]=1$. Observe now that $q\mapsto\widetilde{H}_{q}^{1,\uparrow}(\omega)$ is  non-negative and non-decreasing over $D$ for all $\omega\in\widetilde{\Omega}_{0}$. Now, the $\mathbb{F}$-stopping time $\sigma$, defined as
    \begin{align*}
        \sigma\coloneqq\inf\bigg\{r\in\mathbb{Q}\colon\sup_{q\leq r,\>q\in D}\widetilde{H}_{q}^{1,\uparrow}=\infty\bigg\},
    \end{align*}
    satisfies $P[\sigma=\infty]=1$, and we define
    \begin{align}\label{prop:pr:conv}
        \bar{H}_{t}^{1,\uparrow}(\omega)=\lim_{\substack{q\downarrow\downarrow t,\, q\in\mathbb{Q}}}\widetilde{H}_{q}^{1,\uparrow}(\omega)\quad\text{if}\quad 0\leq t<\sigma(\omega),\qquad\bar{H}_{t}^{1,\uparrow}(\omega)=0\quad\text{if}\quad t\geq\sigma(\omega),
    \end{align}
    for all $t\in[0,T)$, and
    \begin{align*}
        \bar{H}_{T}^{1,\uparrow}(\omega)=\widetilde{H}_{T}^{1,\uparrow}(\omega)\quad\text{if}\quad T<\sigma(\omega),\qquad\bar{H}_{T}^{1,\uparrow}(\omega)=0\quad\text{if}\quad T\geq\sigma(\omega).
    \end{align*}
    Note that by Lemma~\ref{lem:meas} and the right continuity of the filtration, the process $\bar{H}^{1,\uparrow}$ obtained in this way is right continuous and $\mathbb{F}$-adapted.
    Indeed, for fixed $t\in[0,T)$, we have that $\liminf_{\substack{q\downarrow\downarrow t,\, q\in\mathbb{Q}}}\widetilde{H}_{q}^{1,\uparrow}$ and $\limsup_{\substack{q\downarrow\downarrow t,\, q\in\mathbb{Q}}}\widetilde{H}_{q}^{1,\uparrow}$ are $\mathcal{F}_t$-measurable by the right-continuity of the filtration. By Lemma~\ref{lem:meas}, we have
    \begin{align*}
        F:={}&\left\{\omega\in\Omega: \bar{H}_{t}^{1,\uparrow}(\omega)=\lim_{\substack{q\downarrow\downarrow t,\, q\in\mathbb{Q}}}\widetilde{H}_{q}^{1,\uparrow}(\omega)\right\}\\
        ={}&\left\{\omega\in\Omega: \liminf_{\substack{q\downarrow\downarrow t,\, q\in\mathbb{Q}}}\widetilde{H}_{q}^{1,\uparrow}(\omega)=\limsup_{\substack{q\downarrow\downarrow t,\, q\in\mathbb{Q}}}\widetilde{H}_{q}^{1,\uparrow}(\omega)\in\mathbb{R}\right\}\in\mathcal{F}_t,
    \end{align*}
    and hence $P[F]=1$, since $\bar{\Omega}_{0}\coloneqq\widetilde{\Omega}_0\cap\{\sigma=\infty\}$ satisfies  $\bar{\Omega}_{0}\subseteq F$ and $P[\bar{\Omega}_{0}]=1$.

    We now claim that if $(\omega,t)\in\bar{\Omega}_{0}\times(0,T)$ is such that $t$ is a continuity point of the function $s\mapsto \bar{H}_{s}^{1,\uparrow}(\omega)$, then $\widetilde{H}_{t}^{1,n,\uparrow}(\omega)\to \bar{H}_{t}^{1,\uparrow}(\omega)$. Indeed, for $\varepsilon>0$ let $q_{1}<t<q_{2}$ be rational numbers such that $\bar{H}_{q_{2}}^{1,\uparrow}(\omega)-\bar{H}_{q_{1}}^{1,\uparrow}(\omega)<\varepsilon$. From \eqref{eq:komlos_conv_down}, there exists $N=N(\omega)\in\mathbb{N}$ such that
    \begin{align*}
        \big\vert\widetilde{H}_{q_{1}}^{1,n,\uparrow}(\omega)-\bar{H}_{q_{1}}^{1,\uparrow}(\omega)\big\vert<\varepsilon,\quad\big\vert\widetilde{H}_{q_{2}}^{1,n,\uparrow}(\omega)-\bar{H}_{q_{2}}^{1,\uparrow}(\omega)\big\vert<\varepsilon,\quad\forall n\geq N.
    \end{align*}
    We then estimate, for all $n\geq N$,
    \begin{align*}
        \big\vert\widetilde{H}_{q_{2}}^{1,n,\uparrow}(\omega)-\widetilde{H}_{q_{1}}^{1,n,\uparrow}(\omega)\big\vert&\leq\big\vert\widetilde{H}_{q_{2}}^{1,n,\uparrow}(\omega)-\bar{H}_{q_{2}}^{1,\uparrow}(\omega)\big\vert+\big\vert \bar{H}_{q_{2}}^{1,\uparrow}(\omega)-\bar{H}_{q_{1}}^{1,\uparrow}(\omega)\big\vert\\[0.5em]
        &\quad+\big\vert \bar{H}_{q_{1}}^{1,\uparrow}(\omega)-\widetilde{H}_{q_{1}}^{1,n,\uparrow}(\omega)\big\vert\\[0.5em]
        &< 3\varepsilon.
    \end{align*}
    Therefore, using monotonicity of $\widetilde{H}^{1,n,\uparrow}$, we obtain for all $n\geq N(\omega)$ that
    \begin{align*}
        \big\vert\widetilde{H}_{t}^{1,n,\uparrow}(\omega)-\bar{H}_{t}^{1,\uparrow}(\omega)\big\vert&\leq\big\vert\widetilde{H}_{t}^{1,n,\uparrow}(\omega)-\widetilde{H}_{q_{2}}^{1,n,\uparrow}(\omega)\big\vert+\big\vert\widetilde{H}_{q_{2}}^{1,n,\uparrow}(\omega)-\bar{H}_{q_{2}}^{1,\uparrow}(\omega)\big\vert\\[0.5em]
        &\quad+\big\vert \bar{H}_{q_{2}}^{1,\uparrow}(\omega)-\bar{H}_{t}^{1,\uparrow}(\omega)\big\vert\\[0.5em]
        &\leq\big\vert\widetilde{H}_{q_{1}}^{1,n,\uparrow}(\omega)-\widetilde{H}_{q_{2}}^{1,n,\uparrow}(\omega)\big\vert+\big\vert\widetilde{H}_{q_{2}}^{1,n,\uparrow}(\omega)-\bar{H}_{q_{2}}^{1,\uparrow}(\omega)\big\vert\\[0.5em]
        &\quad+\big\vert \bar{H}_{q_{2}}^{1,\uparrow}(\omega)-\bar{H}_{q_{1}}^{1,\uparrow}(\omega)\big\vert\\[0.5em]
        &< 5\varepsilon.
    \end{align*}
    For $t=T$, the convergence of $\widetilde{H}_{T}^{1,n,\uparrow}(\omega)\to \bar{H}_{T}^{1,\uparrow}(\omega)$ follows from \eqref{eq:komlos_conv_down} and the identity $\bar{H}_{T}^{1,\uparrow}=\widetilde{H}_{T}^{1,\uparrow}$ on $\bar{\Omega}_{0}$.

    The process $\bar{H}^{1,\uparrow}$ is not yet the desired limit because we still have to ensure the convergence at the jumps times of $\bar{H}^{1,\uparrow}$. Since $\bar{H}^{1,\uparrow}$ is right continuous and $\mathbb{F}$-adapted, there exists a sequence $(\tau_{k})_{k\in\mathbb{N}}$ of $[0,T]\cup\{\infty\}$-valued $\mathbb{F}$-stopping times exhausting the jumps (and the ``explosions'') of the process $\bar{H}^{1,\uparrow}$. This argument uses Theorem~D together with Remark~E in the preliminary section ``Complements to Chapter~IV'' of \cite{dellacherie1982}.
    Hence, by passing once more to convex combinations, we may also assume that $(\widetilde{H}_{\tau_{k}}^{1,n,\uparrow})$ converges almost surely on $\{\tau_{k}\leq T\}$ for every $k\in\mathbb{N}$. We can therefore set
    \begin{align}
        \widehat{\Omega}_{0}\coloneqq\Big\{\omega\in\bar{\Omega}_{0}\colon\widetilde{H}_{\tau_{k}(\omega)}^{1,n,\uparrow}(\omega)\text{ converges to a limit in $\mathbb{R}$ for all }k\Big\},\label{prop:pr:conv2}
    \end{align}
    and still have $P[\widehat{\Omega}_{0}]=1$. In particular, we again apply Lemma~\ref{lem:komlos} together with a diagonalisation procedure, but this time on the already constructed sequence $(\widetilde{H}^{1,n,\uparrow})$ and for the countably many stopping times $(\tau_{k})_{k\in\mathbb{N}}$. By still denoting this sequence as $(\widetilde{H}^{1,n,\uparrow})$, this construction yields a subset $\widehat{\Omega}_{0}\subseteq\bar{\Omega}_{0}$ of full measure, where the convergence of $\widetilde{H}_{t}^{1,n,\uparrow}$ to a limit in $\mathbb{R}$ holds true for all $t\in[0,T]$. Indeed, for $\omega\in\widehat{\Omega}_{0}$, the convergence in \eqref{prop:pr:conv2} together with \eqref{prop:pr:conv} implies that $\widetilde{H}_{t}^{1,n,\uparrow}(\omega)$ converges to a limit in $\R$ for all $t\in[0,T].$ Finally, we define $H^{1,\uparrow}$ by setting $H^{1,\uparrow}_t(\omega)=\lim_{n\to\infty}\widetilde{H}_{t}^{1,n,\uparrow}(\omega)$ on the set
    $$G:=\left\{(\omega,t)\in\Omega\times[0,T]:\text{ $\widetilde{H}^{1,n,\uparrow}_t(\omega)$ converges to a limit in $\mathbb{R}$}\right\},$$
    and $H^{1,\uparrow}_t(\omega)=0$ on $G^c$. The set $\Omega_{G}\coloneqq\{\omega\in\Omega\colon\forall t\in[0,T],\>(\omega,t)\in G\}$ has full measure, since $\widehat{\Omega}_{0}\subseteq\Omega_{G}$ and $P[\widehat{\Omega}_{0}]=1$. This procedure yields an $\mathbb{F}$-predictable process $H^{1,\uparrow}$ by Lemma~\ref{lem:meas}, since the processes $\widetilde{H}^{1,n,\uparrow}$ are $\mathbb{F}$-predictable. Moreover, since for $\omega\in\widehat{\Omega}_{0}$ we have that $\widetilde{H}^{1,n,\uparrow}_t(\omega)$ converges for all $t\in[0,T]$, and the mapping $t\mapsto\widetilde{H}^{1,n,\uparrow}_t(\omega)$ is non-negative and non-decreasing for all $n$, we have that $t\mapsto H^{1,\uparrow}_t(\omega)$ is non-negative and non-decreasing for all $\omega\in\Omega_0\subseteq\widehat{\Omega}_0$ with $P[\Omega_0]=1$. Furthermore, it holds that $P[H_{T}^{1,\uparrow}<\infty]=1$.

    The case $H^{1,\downarrow}$ is treated analogously. In particular, we obtain two processes, $H^{1,\uparrow}$ and $H^{1,\downarrow}$, both $[0,\infty]$-valued, $\mathbb{F}$-predictable, $P$-a.s.~increasing, satisfying $P[H_{T}^{1,\uparrow}<\infty]=1$ and $P[H_{T}^{1,\downarrow}<\infty]=1$,
    such that
    \begin{align*}
        P\big[\widetilde{H}_{t}^{1,n,\uparrow}\to H_{t}^{1,\uparrow},\ \forall t\in[0,T]\big]=1,\quad\text{and}\quad P\big[\widetilde{H}_{t}^{1,n,\downarrow}\to H_{t}^{1,\downarrow},\ \forall t\in[0,T]\big]=1.
    \end{align*}

    To conclude the proof, define the process $H^{1}=(H_{t}^{1})_{0\leq t\leq T}$ as
    \begin{align}\label{eq:limit_process_H1}
        H_{t}^{1}(\omega)=H_{t}^{1,\uparrow}(\omega)\mathbbm{1}_{\llbracket 0,\sigma_{1}\wedge\sigma_{2}\llbracket}(\omega,t)\mathbbm{1}_{\llbracket 0,\sigma_{3}\rrbracket}(\omega,t)-H_{t}^{1,\downarrow}(\omega)\mathbbm{1}_{\llbracket 0,\sigma_{1}\wedge\sigma_{2}\llbracket}(\omega,t)\mathbbm{1}_{\llbracket 0,\sigma_{3}\rrbracket}(\omega,t),
    \end{align}
    where the stopping times $\sigma_{1}$, $\sigma_{2}$ and $\sigma_{3}$ are defined as
    \begin{align*}
        \sigma_{1}&\coloneqq\inf\{t>0\colon H_{t-}^{1,\uparrow}=\infty\text{ or }H_{t-}^{1,\downarrow}=\infty\},\\
        \sigma_{2}&\coloneqq\inf\{t>0\colon \Delta H_{t}^{1,\uparrow}=\infty\text{ or }\Delta H_{t}^{1,\downarrow}=\infty\},\\
        \sigma_{3}&\coloneqq\inf\{t>0\colon \Delta_{+}H_{t}^{1,\uparrow}=\infty\text{ or }\Delta_{+}H_{t}^{1,\downarrow}=\infty\}.
    \end{align*}
    The stopping times $\sigma_{1}$ and $\sigma_{2}$ are predictable (see Remark~IV.87(d) in \cite{dellacherie1978} together with Remark~E of the preliminary section ``Complements to Chapter~IV'' in \cite{dellacherie1982}). Hence, the process $H^{1}$ is real-valued and $\mathbb{F}$-predictable. Moreover, $H^{1,\uparrow}$ and $H^{1,\downarrow}$ both are $P$-a.s.~increasing, and because $P[H_{T}^{1,\uparrow}<\infty]=1$ as well as $P[H_{T}^{1,\downarrow}<\infty]=1$, it holds that $P[\sigma_{1}=\infty]=P[\sigma_{2}=\infty]=P[\sigma_{3}=\infty]=1$. The total variation $\vert H^{1}\vert=(\vert H^{1}\vert_{t})_{0\leq t\leq T}$ of $H^{1}$ satisfies $\vert H^{1}\vert_{t}<\infty$ for all $\{0\leq t<\sigma_{1}\wedge\sigma_{2}\}\cap\{0\leq t\leq\sigma_{3}\}$ and $\vert H^{1}\vert_{t}=\infty$ for all $\{t\geq\sigma_{1}\wedge\sigma_{2}\}\cup\{t>\sigma_{3}\}$, which implies that $\vert H^{1}\vert$ is $[0,\infty]$-valued, $\mathbb{F}$-predictable and $P$-a.s.~finite.
    It remains to check that $H^{1}\in\mathcal{A}(x)$. By construction, the processes $(\widetilde{H}^{1,n})_{n\in\mathbb{N}}$, defined via the decomposition $\widetilde{H}^{1,n}=\widetilde{H}^{1,n,\uparrow}-\widetilde{H}^{1,n,\downarrow}$, are $x$-admissible for all $\theta\in\Theta$. In particular, the process $H^{1}$ and the sequence $(\widetilde{H}^{1,n})_{n\in\mathbb{N}}\subseteq\mathcal{A}(x)$ satisfy \eqref{eq:conv_H1_decomp}, so Lemma~\ref{lem:conint} implies that $H^{1}\in\mathcal{A}(x)$. This finishes the proof.
\end{proof}
\end{appendices}


\end{document}